\documentclass[11pt]{article}
\usepackage[utf8]{inputenc}
\usepackage{fullpage}
\usepackage[margin=1in, letterpaper]{geometry}
\usepackage{framed}
\usepackage{my_package}
\usepackage{multirow}
\usepackage{microtype}
\usepackage{XCharter}
\usepackage{dirtytalk}
\usepackage{paralist}
\usepackage{thm-restate}
\usepackage{tcolorbox}
\usepackage{complexity}
\usepackage{enumitem}

\definecolor{WhiteSmoke}{rgb}{0.96, 0.96, 0.96}
\definecolor{shadecolor}{named}{WhiteSmoke}

\newcommand{\MLS}{\textsc{Ordered $\alpha$-Local }\allowbreak\textsc{Search}}
\newcommand{\HMLS}{\textsc{Hybrid Ordered $\alpha$-Local }\allowbreak\textsc{Search}}

\newcommand{\ow}{{\bar{w}}}

\newcommand{\lt}{<}
\newcommand{\gt}{>}

\newcommand{\Anear}{A^{\textrm{near}}}

\newcommand{\Dnear}{D^{\textrm{near}}}
\newcommand{\Dfar}{D^{\textrm{far}}}

\newcommand{\Orem}{O^{\textrm{rem}}}

\newcommand{\Sfarunst}{S^{\textrm{far-unstable}}}
\newcommand{\Sfarst}{S^{\textrm{far-stable}}}
\newcommand{\Snear}{S^{\textrm{near}}}
\newcommand{\btheta}{\bar\theta}

\title{Submodular Maximization over Many Matroids via Ordered Local Search}

\author{
Neta Singer \thanks{Ecole Polytechnique Fédérale de Lausanne
(Email: \href{mailto:neta.singer@epfl.ch}{neta.singer@epfl.ch})}
\and Theophile Thiery \thanks{Department of Computer Science, ETH Zurich, Zurich, (Email:
\href{mailto:theophile.thiery@inf.ethz.ch}{theophile.thiery@inf.ethz.ch})}
}

\begin{document}
\maketitle
\pagestyle{plain}
\begin{abstract}
    Given a monotone submodular function, we consider the problem of finding a maximum-valued set in the intersection of $k$ matroids.
    Our main result is a polynomial time local search based algorithm achieving a $\frac{k}{2} + o(k)$ approximation guarantee.
    This asymptotically matches the best-known guarantee of $\frac{k}{2} + \e$ in the unweighted setting by Lee, Sviridenko, and Vondrák (2009). 
    Prior to this work, the state-of-the-art was a $\frac{\ln(4)k}{1+\ln(2)} + o(k)$-approximation algorithm obtained by Feldman and Ward (2026).
    Our approach extends to Matroid $k$-Parity yielding the same approximation guarantee.

    In contrast to the weight bucketing approach underlying the recent advances of Singer and Thiery (2025) and Feldman and Ward (2026), our algorithm processes elements greedily in decreasing order of marginal value and searches for sufficiently profitable swaps, whose gain exceeds a parameter $\alpha$ given as a function of $k$. 
    %This approaches a $k/2$-approximation for submodular maximization over a matroid intersection as $k$ grows.
    We further combine this idea with the weight bucketing approach to obtain improved guarantees for weighted $k$-Set Packing. Our second main result is a $\frac{\ln(4)k}{3} + o(k)$-approximation algorithm for weighted $k$-Set Packing, improving on the state of the art $\frac{k}{2.00561} + O(1)$-approximation by Neuwohner (2023).
\end{abstract}
\newpage
\tableofcontents
\newpage

    \section{Introduction}

    Maximizing submodular and linear functions subject to a $k$-matroid intersection constraint is a central topic in combinatorial optimization. 
    When $k = 1$, \textsc{Greedy} is optimal for maximizing linear functions, and celebrated algorithms by Calinescu, Chekuri, P\'al, and Vondr\'ak \cite{Calinescu-Chekuri-Pal-Vondrak:2011:Maximizing} and Filmus and Ward \cite{Filmus:2014:Monotone, Buchbinder-Feldman:2024:Deterministic} are optimal for monotone submodular functions.
    For $k = 2$, Edmonds' matroid intersection algorithm outputs an optimal solution for linear objective functions \cite{Edmonds:2003:Submodular}.
    Such algorithms imply the existence of optimal methods for finding maximum weight spanning tree, matchings, disjoint arborescences, and for maximizing welfare in the allocation of goods to agents since they are special cases of linear and submodular maximization over matroid constraints \cite{Lee-Ryan:1992:Matroid, Calinescu-Chekuri-Pal-Vondrak:2011:Maximizing}. \\

    In this paper, we consider a $k$-Matroid Intersection constraint for $k \geq 3$. Until recent years, no multiplicative improvement over \textsc{Greedy}'s guarantee was known even for linear objectives beyond cardinality maximization.
    \textsc{Greedy} achieves a $k$-approximation for linear functions, and $(k+1)$-approximation for monotone submodular functions \cite{Nemhauser-Wolsey-Fisher:1978:Analysis}.
    A linear dependence on $k$ is needed. Hazan, Safra, and Schwartz \cite{Hazan-Safra-Schwartz:2006:Complexity} first showed that the problem is $\NP$-hard to approximate within a factor of $\Omega(k/\log(k))$.
    This was recently strengthened by Lee, Svensson, and Thiery \cite{Lee:2024:Asymptotically}, who obtained a $k/12$-hardness result, and further by Minzer and Zheng \cite{Minzer-Zhe_Zheng:2025:Near-Optimal} to $k/8$, assuming $\P \neq \NP$. \\

    For cardinality maximization, there exists a $\frac{k}{2} + \e$-approximation algorithm for any constant $\e \gt 0$ \cite{Lee:2010:Matroid}.
    The general linear case is in fact more complex since global counting arguments, which are used in the cardinality case, are no longer sufficient.
    Singer and Thiery recently overcame this barrier and presented the first multiplicative improvement over \textsc{Greedy} \cite{Singer-Thiery:2025:Better} for linear objectives.
    They designed a $\frac{k}{\ln(4)} + O(1)$-approximation algorithm improving over the former $(k-1)$-approximation algorithm by Lee, Sviridenko, and Vondrák \cite{Lee:2010:Submodular}. Building on these ideas, Feldman and Ward \cite{Feldman-Ward:2026:Submodular} extended the framework to monotone submodular maximization, obtaining approximation ratios of $k\ln(2)+O(1)$ and $\frac{k\ln(4)}{\ln(2)+1}+o(k)$ for linear and monotone submodular objectives, respectively.
    Despite this progress, these works leave open the question of bridging the gap between the unweighted and the weighted/submodular settings. In this work, we close this gap fully up to additive terms.
    \begin{restatable}{mythm}{thmone}
        \label{thm:main-approximation-guarantee}
        There exists a $\frac{k}{2} + o(k)$-approximation algorithm for maximizing a monotone submodular function subject to a $k$-matroid intersection constraint.
    \end{restatable}
    Our algorithm is the first to achieve an approximation ratio for monotone submodular functions that asymptotically matches the best unweighted approximation guarantee \cite{Lee:2010:Matroid}. We remark that, similarly to \cite{Lee:2010:Matroid, Singer-Thiery:2025:Better, Feldman-Ward:2026:Submodular} our guarantees hold for the more general Matroid $k$-Parity constraint. This leads to improved guarantees for \emph{Bounded Degree Maximum Independent Set} considered recently by Bai, B{\'{e}}rczi, and Siemelink \cite{Bai-Berczi-Siemelink:2025:Approximating} which can be cast as a Matroid $k$-Parity constraint for example. Further applications can be found in \cite{Cheung-Lau-Leung:2014:Algebraic}.
    This result settles the gap for large values of $k$ and leaves as exciting open lines of research the question of finding a matching guarantee for low values of $k$, and of improving over the unweighted guarantees. \\

    Notably, our algorithm is a simple local search-based method. However, we make a crucial observation that the local search performs better when $(1)$ local swaps are processed in decreasing order over the edges' weight (marginal contribution in the submodular setting), and $(2)$ only swaps that make significant gain, as parametrized by a function in $k$, are taken. Variants of this greedy processing of the local search order have been employed in \cite{Chandra-Halldorsson:1999:Greedy, Singer-Thiery:2025:Better, Feldman-Ward:2026:Submodular}.\\

    We further leverage our techniques to obtain an improved approximation ratio for $k$-Set Packing.
    This problem is a special case of $k$-Matroid Intersection with a long line of research gradually improving on approximation guarantees \cite{Hurkens:1989:Size, Cygan:2013:Improved,Arkin:1998:Local,Halldorson:1995:Approximating, Chandra-Halldorsson:1999:Greedy, Berman:2000:Approximation, Berman-Krysta:2003:Optimizing, Neuwohner:2021:Improved, Neuwohner:2022:Limits,Neuwohner:2023:Passing, Thiery:2023:Improved}. For this special case, better approximation ratios than $\frac{k}{2}$ were already known in the weighted setting. The best guarantee for weighted $k$-Set Packing is a $\frac{k}{2.00561} + O(1)$-approximation due to Neuwohner \cite{Neuwohner:2023:Passing}. Her algorithm employs Cygan's \cite{Cygan:2013:Improved} $\frac{k+1}{3}$-approximation algorithm for the unweighted problem to overcome the $\frac{k}{2}$-barrier.
    In this work, we combine our asymptotic $\frac{k+o(k)}{2}$-approximation algorithm with the algorithm of Singer and Thiery \cite{Singer-Thiery:2025:Better} (using the existence of a $\frac{k+1}{3}$-approximation for the unweighted problem) to improve upon Neuwohner's result \cite{Neuwohner:2023:Passing}. We achieve an improved approximation guarantee of $\frac{k}{2.16404} + o(k)$ for linear objectives, and our method provides an alternative to those based on Berman's $w^2$ local search algorithm \cite{Berman:2000:Approximation,Neuwohner:2021:Improved, Neuwohner:2022:Limits, Thiery:2023:Improved, Neuwohner:2023:Passing}.
    \begin{mythm}
        \label{thm:intro-2}
        There exists a $\frac{\ln(4)k}{3}+ o(k)$-approximation algorithm for weighted $k$-Set Packing.
    \end{mythm}

    \subsection{Overview of our Approach}
    Our main algorithm is inspired by the method introduced in \cite{Singer-Thiery:2025:Better} and continued in \cite{Feldman-Ward:2026:Submodular}. These advances interpolate between greedy and local search.
    They randomly partition the instance into almost unweighted weight buckets, process them greedily, and run the state-of-the-art unweighted algorithm on each weight class.
    Unfortunately, this approach suffers from the inherent tradeoff between the approximation obtained in a single weight bucket and the loss incurred by committing to the greedy order over the buckets. \\

    In contrast, our algorithm avoids bucketing altogether and processes the edges in decreasing order of their weight in a \emph{continuous} fashion. Concretely, we perform a local-search parametrized by values $\alpha$, and $\theta$. The algorithm processes all edges of \textit{weight} at least $\theta$ in decreasing order of $\theta$ and makes all local swaps that increase the value by at least $\alpha$. An $\alpha$-swap adds a set of edges whose weight is at least $\alpha$ times that of those removed. Thus when no $\alpha$-swap can be found, the threshold $\theta$ is decreased so edges of smaller weight become eligible for local search.
    When the objective function is submodular, edge \textit{weights} are defined by an auxiliary weight function akin to \cite{Feldman-Ward:2026:Submodular}.

    \subsection{Techniques}
    We now give some intuition as to why the above method achieves a $\frac{k}{2} + o(k)$-guarantee and for simplicity assume we work with a $k$-set packing instance and a linear function. Typical for local search methods, we will construct a sequence of non-improving local swaps, and aim to aggregate them to derive an approximation guarantee. 
    In particular, we consider $O_\theta$ to be the edges of $\opt$ of weight $\theta$ and consider our algorithm's solution $A_\theta$ after processing all edges of weight at least $\theta$. 
    A typical unweighted analysis shows that the weight of the neighborhood of $O_\theta$ in $A_\theta$ is at least a $\frac{2}{k+\alpha}$ portion of the weight $O_\theta$. We then consider two sets: (1) the edges $o \in O_\theta$ whose neighborhood remains in the final output after processing all weights, or (2) the edges $o \in O_\theta$ whose neighborhood gets discarded in some local search iteration preceding the final output. \\
    
    Aggregating over edges for which the first case happens, we can intuitively recover the $\frac{k + \alpha}{2}$-guarantee. This guarantee is successful for small enough $\alpha$.
    In the second case, each discarded neighbor implies the existence of an $\alpha$-swap greatly increasing the value of our algorithm's solution. We prove that the weight of $O$ for which (2) happens is bounded by $\frac{k}{\alpha - 1}f(A)$, where $A$ is the final solution. This fact is reminiscent of the analyses of \cite{Chandra-Halldorsson:1999:Greedy, Chekuri-Gupta-Quanrud:2015:Streaming}. This requires $\alpha$ to be a large parameter to minimize the effect of (2).
    Combining these two facts proves our main theorem when $\alpha = \theta(\sqrt{k\log(k)})$. 

        \paragraph{Dealing with Matroid $k$-Parity:} There is a technical complication when dealing with a matroid $k$-parity constraint. Indeed, while the choice of neighbors is direct for $k$-Set Packing as hyperedges that cover the same vertices, our proof proceeds by fractionally assigning neighborhoods via the matroid independence oracle. 
        More precisely, consider some $o \in O$ such that $w_o = \theta_0$. Then, we can decompose $\theta_0 = \sum_{i} (\theta_{i-1} - \theta_i)$ for all thresholds $\theta_i$ considered by the algorithm. Then, we assign a $\theta_{i-1}  - \theta_i$ portion of $o$'s weight to its neighborhood in the round that threshold $\theta_i$ is introduced. 
        We therefore do not need to explicitly keep track of neighborhoods and can instead naturally integrate over $\theta$.

        \paragraph{Dealing with Submodularity:}
        When $f$ is a submodular set function, edges no longer have a fixed weight.  This is handled by Feldman and Ward \cite{Feldman-Ward:2026:Submodular} by processing edges in order of decreasing marginal contribution with respect to the algorithm's current solution. They further maintain an ordering $\prec$ of the ground set to define auxiliary weights of edges in the algorithm's solution.
        In the method of \cite{Feldman-Ward:2026:Submodular}, the marginal contribution of an edge can only decrease over time since the solution is built greedily.
        However, our algorithm does not have this property, as edges can enter, exit, and reenter our solution over time.
        For this reason, we maintain an ordered set $(X, \prec)$ of edges that were taken into the solution at any iteration of our algorithm. We then evaluate the marginal contribution of an edge with respect to $X$, a superset of all the current solution. A local $\alpha$-swap is then defined by comparing the auxiliary weight of the added and discarded edges.
        
        \paragraph{Weighted $k$-Set Packing.} In the special variant of weighted $k$-set packing, we go beyond $k/2$ and obtain a $\frac{\ln(4)k}{3}$-approximation. To obtain this approximation, we combine the algorithm of \cite{Singer-Thiery:2025:Better} with our $\alpha$-local search as follows: we construct a random partitioning of edges into weight buckets as in \cite{Singer-Thiery:2025:Better}, then we process the weight buckets in decreasing weight order, applying the state-of-the-art $k/3$-approximation algorithm \cite{Cygan:2013:Improved}. However, before processing the next weight bucket, we run our $\alpha$-local search over the entire instance restricted to the weight buckets processed so far. This additional $\alpha$-local search step removes inherent flaws in the greedy processing of weight buckets, wherein higher weight edges can block up to $k$-many optimal edges of slightly lower weight. However, the weight bucketing also guarantees that the approximation goes beyond $k/2$ by interpolating with the $k/3$-guarantee of running the unweighted algorithm.

 \subsection{Paper Organization}

The rest of the paper is organized as follows. In \Cref{sec: prelim} we give necessary notations and definitions regarding matroids and submodular functions. In \Cref{sec: algo submod} we define the $(\theta, \alpha)$-Local Search algorithm and our main submodular maximization algorithm, which iteratively runs $(\theta, \alpha)$-Local Search over decreasing $\theta$. 
In \Cref{sec: analysis of algo 2} and \Cref{sec: matroid swaps} we give preliminaries for the analysis of our submodular maximization algorithm by constructing the necessary matroid exchanges with respect to $\opt$. In \Cref{sec: approx analysis submod} we complete the algorithm analysis aggregating the matroid exchanges into approximation guarantees with respect to the final output. In \Cref{sec: set pack} we define the extended algorithm for weighted $k$-Set Packing and analyze it.

    \section{Preliminaries}\label{sec: prelim}
        For any integer $k \in \NN$, we use the convention $[k] \triangleq \{1, \ldots, k\}$. We assume that the reader is familiar with basic matroid definitions (see \cite{Schrijver:2003:Combinatorial} for a good survey). We introduce the combinatorial constraints that we study in this paper. 

        \paragraph{Matroid $k$-parity:}  We are given a hypergraph $G = (V, E)$ on a set of vertices $V$ and a matroid $\cM = (V, \cI)$ defined on $V$. Each hyperedge $e \in E$ has contains up to $k$ vertices.
        Given a set function $f\colon 2^E \to \mathbb{R}$, the goal is to find a maximum value collection of disjoint hyperedges $M \subseteq E$ such that the vertices incident to $M$ are independent in $\cI$. For any collection of hyperedges $S \subseteq E$, we say $S \in \mathcal{I}$ if $v(S) \in \mathcal{I}$, i.e. the vertices covered by the hyperedges are independent. 
        As shown by Lee, Sviridenko and Vondr\'ak, we may assume that the hyperedges are disjoint (modifying $\cM$) \cite{Lee:2010:Matroid}.

        \paragraph{$k$-Set Packing and $k$-Matroid Intersection:} $k$-Set Packing is a special case of Matroid $k$-Parity, where $\cM$ is the \emph{free} matroid ($\cI = 2^V$). 
        Matroid $k$-Parity also generalizes $k$-Matroid Intersection  which consists of finding a maximum value independent set in $k$ matroids $(E, \cI_1), \ldots, (E, \cI_k)$. The reduction to Matroid $k$-parity is standard by creating $k$ copies $(e, i)_{i \in [k]}$ for each $e \in E$ with a unique hyperedge incident to these $k$ copies, and defining a set $S$ independent in $\cM$ if $\{e \in E \colon (e, i) \in S\} \in \cI_i$ for all $i \in [k]$.

        \paragraph{Set functions:}
            We let $f: 2^E \rightarrow \mathbb{R}$ be a function that assigns a value to every subset of $E$. Given a set $S$ and an element $e$, we define the \emph{marginal contribution} of $e$ with respect to $S$ as $f(e \mid S) \triangleq f(S \cup \{ e \}) - f(S)$.
            Similarly, given a set $T \subseteq E$, we denote by $f(T \mid S) \triangleq f(S \cup T) - f(S)$ the marginal contribution of $T$ to $S$. 
            A set function $f$ is \emph{monotone} if $f(e \mid S) \geq 0$ for every set $S \subseteq E$ and element $e \in E$.
            A set function $f$ is called \emph{submodular} if $f(e \mid S) \geq f(e \mid T)$ for every two sets $S \subseteq T$ and element $e \in E \setminus T$. In other words, $f$ is submodular if the marginal contribution of an element $e$ w.r.t. $S$ can diminish when additional elements are added to the set $S$. 
            A special case of monotone submodular functions is when each element $e \in E$ is assigned a weight $w_e \geq 0$ and $f(S) = \sum_{e \in S} w_e$ for all $S \subseteq E$, and $f(\emptyset) = 0$. In this case, we say that $f$ is \emph{linear}.
            Our algorithms optimize a submodular set function $f$ over a constraint involving matroids. Since both $f$ and $\cM$ may have an exponential size description, we make the standard assumption that $f$ is accessed through a \emph{value oracle} that given a set $S \subseteq E$ returns $f(S)$, and that $\cM$ is accessed via a \emph{independence oracle} that given $S \subseteq E$ returns yes if $v(S) \in \Ical$. \\

         \textbf{Other Notations:} We denote by $W \triangleq \max_{e \in E\colon e\in \cI} f(e)$ the maximum weight of any feasible edge and denote an interval by $I \subseteq [0, W]$. Additionally, given a collection of sets $\{R_i\}_{i = 1}^{L}$, we will denote $R_{\leq i} = \bigcup_{j = 1}^{i} R_j$, and $R_{\geq i} = \bigcup_{j = i}^{L} R_j$.

    \section{Main Algorithm}
    \label{sec: algo submod}
    We introduce our main algorithm for maximizing a monotone submodular function subject to a matroid $k$-parity constraint, described in \Cref{alg:merge-LS}.
    It uses as a subroutine a local-search algorithm (\Cref{alg:local-search}) parametrized by two values $\alpha$, and $\theta$. The parameter $\alpha$ ensures that the ratio of value of edges added to value of edges removed in the solution is at least $\alpha$.
    The parameter $\theta$ restricts the local search to edges whose marginal value is at least $\theta$.
    \Cref{alg:local-search} starts by applying this subroutine with threshold $\theta = W$. The threshold is then decreased continuously while making $\alpha$-improving swaps whenever one can be made.

    \subsection{Defining the $\alpha$-Local Search Subroutine}
    To define this algorithm, we introduce the \emph{incremental value} of an edge and an ordering $\prec$ of the edges of $E$.
    This notion has appeared for example in \cite{Ward:2012:Approximation,Chekuri-Gupta-Quanrud:2015:Streaming, Feldman-Ward:2026:Submodular} as a way to run local search methods over submodular objectives.
    \begin{mydef}[$\prec$-Incremental value]
        \label{def:incremental-value}
        Given a submodular function $f \colon 2^{E} \rightarrow \RR_+$, and an ordering $\prec$ of $E$. We define the $\prec$-incremental value of $e \in E$ with respect to $S \subseteq E$ as:
        \begin{align*}
            w_{\prec}(e, S) & \triangleq f(e \mid \{e' \in S \colon e' \prec e\})\enspace.
        \end{align*}
        Given $S \subseteq E$, we define $S_e^{-} \triangleq \{e' \in S \colon e' \prec e\}$ as the set of elements in $S$ that are ordered before $e$.

        For ease of notation, we sometimes omit the subscript $\prec$ when the order is clear from the choice of the set $S$.
    \end{mydef}

    \Cref{alg:merge-LS} maintains an ordering $\prec$ that is updated over time on a set $X \subseteq E$ of edges that were inserted by any prior local search swap, including those that were subsequently removed. The ordering over $E$ is initially empty, and whenever an edge $e \in E$ is inserted in a local search swap, it is placed into the ordering $\prec$ at the end of the ordering (at the end of the queue). 
    In particular, even if $e$ was already ordered by a prior iteration, if $e$ reenters the solution it is placed at the end of $\prec$. 
    Edges will be inserted via $(\theta, \alpha)$-local swaps described below. 
    
    \begin{mydef}[$(\theta, \alpha, X)$-improvement for $A$]
        \label{def:improvement-submod}
        Let $A \subseteq E$ be a feasible solution.
        Given parameters $\theta > 0$ and $\alpha \geq 1$, and an ordered set $(X, \prec)$ such that $A \subseteq X$, we say that a pair of sets $S \subseteq E \setminus A$ and $N \subseteq A$ is a \emph{$(\theta, \alpha, X)$-improvement for $A$} if $(A \cup S) \setminus N \in \cI$, and if the following condition holds:
        \begin{align}
         |S| = \ell \leq k, |N| \leq \ell k \text{ and there exists an ordering }S = \{s_1, \ldots, s_\ell\}\text{ of }S \text{ such that } \notag\\f(s_j \mid X - \{s_j\} + \{s_1, \ldots, s_{j-1}\}) \geq \theta\text{ for all }j \leq \ell\text{, and }f(S \mid X-S) > \alpha \cdot \sum_{e \in N} w_\prec(e, A). \notag 
         \end{align}
        When $X$ is clear from the context, we will write that a pair $(S, N)$ is a $(\theta, \alpha)$-improvement for $A$.
    \end{mydef}

    This defines a swap where we add a set $S$ of edges with individual marginal value at least $\theta$ and where $f(S \mid X - S)$ is at least $\alpha$-times larger than the incremental values of the discarded edges in order to add in $S$.
    \begin{remark}
        \label{rmk:greedy-addition}
        If there is an edge $e \in E$ such that $A + e \in \cI$, and $f(e \mid X - e) > \theta$ then $(e, \emptyset)$ is a $(\theta, \alpha, X)$-improvement for $A$ for any choice of $\alpha \geq 0$. This simply corresponds to the greedy addition of $e$ to $A$.
    \end{remark}
    \begin{remark}
        \label{rmk:update-ordering}
        If a swap adds multiple edges to the current solution, we set the ordering $\prec$ of the newly added edges according to the one given by \Cref{def:improvement-submod}. In particular, we have $s_1 \prec s_2 \prec \ldots \prec s_\ell$.
    \end{remark}
    \begin{remark}
        We can find $(\theta, \alpha)$-improving swaps in polynomial-time when $k$ is constant.
    \end{remark}
        \begin{algorithm}[ht!]
            \DontPrintSemicolon
            \caption{$(\theta, \alpha)$-\LS$(A, X)$}
            \label{alg:local-search}
            \KwIn{Feasible solution $A \subseteq E$, and $X$ such that $A \subseteq X$, ordering $\prec$ of $X$.}
            \While{there exists $(S, N)$ that is $(\theta, \alpha, X)$-improving swap}{
                $A \leftarrow A \setminus N \cup S$ \hfill \tcp{Improve the solution }
                Update $\prec$ by replacing $s_1 \prec \hdots \prec s_l$ at the end \hfill \tcp{Reorder according to \Cref{rmk:update-ordering}}
                $(X, \prec) \leftarrow (X \cup S, \prec)$ \hfill \tcp{Update $X$ and $\prec$ accordingly}
            }
            \Return $(A, X, \prec)$\;
        \end{algorithm}

    \subsection{Defining Ordered $\alpha$-Local Search}
    The conceptual idea behind our main algorithm, $\MLS$ is to run $(\theta, \alpha)$-$\LS$ in order of decreasing $\theta$. The local-search algorithm runs continuously while the threshold $\theta$ is decreased over time, thereby allowing lower-value edges to become eligible.
    Formally, $\MLS$ (\Cref{alg:merge-LS}) uses $(\theta, \alpha)$-$\LS$ to find all $(\theta, \alpha)$-improving swaps. When none can be found, the threshold $\theta$ is decreased to include the next eligible edge.
    We decrease the threshold by looking at the lowest marginal of any edge in an ordered $k$-tuple. 
    If $(A, X)$ is the state of the algorithm, this is defined for $S \subseteq E \setminus A$ as $\Delta_{S} \triangleq \min_{\sigma \in \mathbb{S}_{|S|}} \min_{j \in [|S|]} f(s_{\sigma(j)} \mid X - \{s_{\sigma(j)}\} + \{s_{\sigma(1)}, \ldots, s_{\sigma(j-1)}\})$, where $\mathbb{S}_{\ell}$ is the set of all permutations of size $\ell$.
    We truncate the algorithm by running only until $\theta \geq \btheta \triangleq \delta W/|E|$ such that it halts in polynomial time.
    A standard polynomial run time analysis is shown in Appendix~\ref{sec:runtime}.
    \begin{algorithm}[h!]
            \DontPrintSemicolon
            \SetKw{KwBy}{by}
            \caption{\MLS$(E, \alpha)$}
            \label{alg:merge-LS}
            \KwIn{Parameter $\alpha$ dictates the value of the improvement, $\delta$ accuracy parameter}
            %Let $L \triangleq - \lceil \log_{1- \e}(|E|\delta^{-1})\rceil $ \hfill \tcp{Number of intervals}
            Let $W = \max_{e \in E \colon v(e) \in \cI} f(e)$ \hfill \tcp{Maximum marginal element}
            Let $\theta \gets W$, $\theta' \gets \infty$ \hfill\tcp{Initialize marginal value.}
            Let $A_{\infty} \leftarrow \emptyset$, $X_\infty \leftarrow \emptyset, \prec_\infty \leftarrow \emptyset$ \hfill \tcp{Initialize $A$ and $X$}
            \While{$\theta > \delta \cdot \frac{W}{|E|}$}
            {
            $(A_\theta, X_\theta, \prec_\theta) \gets  (\theta, \alpha)$-$\LS(A_{\theta'}, X_{\theta'}, \prec_{\theta'})$\;
            $\theta' \gets \theta$\;
            $\theta \gets \max_{S \subseteq E \setminus A \colon |S| \leq k}\{\Delta_S \colon \Delta_S \lt \theta'\}$ 
            }
            \Return $(A_{\theta'}, X_{\theta'}, \prec_{\theta'})$\hfill\tcp{Final output at termination condition.}
    \end{algorithm}
    \begin{restatable}{myprop}{runtime}
        \label{lem:runtime-merge-LS}
        \Cref{alg:merge-LS} runs in polynomial time. Moreover, there is only a finite set of thresholds $\theta$ considered by the algorithm.
    \end{restatable}

    \section{Preliminaries for the Analysis of \Cref{alg:merge-LS}}
\label{sec: analysis of algo 2}
To analyze \Cref{alg:merge-LS}, we will let $O$ be an optimal matroid $k$-parity solution. 
Let $\btheta \triangleq \delta \cdot \frac{W}{|E|}$ be the termination threshold. Let $W = \theta_1 > \hdots > \theta_L > 0$ be the finitely many $\theta \in [0, W]$ such that $(\theta, \alpha)$-$\LS$ is called by \Cref{alg:merge-LS}. For every $i \in [1, L]$ and any $\tau \in [\theta_{i+1}, \theta_i]$, we let $(A_\tau, X_\tau) = (A_{\theta_i}, X_{\theta_i})$ where $\theta_{L+1} = 0$. 

\begin{claim}\label{claim:no tau swaps}
    There are no $(\tau, \alpha, X_\tau)$-improving swaps over $A_\tau$ for any $\tau \in [\btheta, W]$.
\end{claim}
\begin{proof}
    For $\tau = \theta_i$, it follows directly from the termination of $(\theta_i, \alpha)$-\LS.
    Let $i \in [L-1]$. Notice that for $\tau \in (\theta_{i+1}, \theta_i)$, there are no $(\tau, \alpha, X_{\theta_i})$-improving swaps over $A_{\theta_i}$ since there is no ordered tuple $(S, \sigma)$ such that $f(s_{\sigma(j)} \mid S - \{s_{\sigma(j)}\} \cup \{s_{\sigma(1)}, \ldots, s_{\sigma(j-1)}\}) \in [\tau, \theta_i) $ by definition of $\Delta$ in \Cref{alg:merge-LS}. In particular, there are no $(\tau, \alpha, X_\tau)$-improving swaps over $A_\tau$ for any $\tau \in [\btheta, W]$.
\end{proof}

Let the final output of \Cref{alg:merge-LS} be denoted by $A \triangleq A_0 = A_{\theta_L}$. We want to compare the value of $f(A)$ to $f(O)$. 
To do this, we define a set $O_\theta$ containing the edges of $O$ of value at least $\theta$.
The value of an element of $O$ is defined by an auxiliary weight $\ow$. 

\subsection{Partitioning the Edges of $O$ by Value}
    Similarly to \cite{Feldman-Ward:2026:Submodular}, we let $O \triangleq \{o_1, \ldots, o_{|O|}\}$ be an arbitrary ordering of $O$. The set of optimal edges that appeared prior to $o_\ell$ in the ordering is denoted by $O^{-}_{o_\ell}$. This ordering can be different from the one computed by \Cref{alg:merge-LS}.
    For every $\theta \in [0, W]$, we define an auxiliary weight at time $\theta$ for each $o \in O$. This auxiliary weight is only useful for the analysis and is not computed by the algorithm.
    For any $o \in O$, we define
    \begin{align*}
        \ow_\theta(o) \triangleq  f(o \mid X_\theta - \{o\} \cup O^{-}_o).
    \end{align*}
    This is the marginal contribution of $o$ with respect to the set of accepted edges up until marginal $\theta$ in the algorithm, as well as the optimal edges prior to $o$ in the ordering.
    Using submodularity, it is easy to see that $\ow_\theta(o)$ is increasing in $\theta$ for any fixed $o \in O$. This is since $X_{\theta} \subseteq X_{\gamma}$ for any $\theta \geq \gamma$.
    \begin{claim}
        \label{claim:non-increasing-auxiliary-weight}
        For any $o \in O$, its auxiliary weight $\ow_\theta(o)$ is increasing as a function of $\theta \in [0, W]$.
    \end{claim}

    We partition $O$ according to these auxiliary weights. Starting at $\theta = W$ and then decreasing $\theta$ to $0$, we define recursively:
    \begin{align}
        O_\theta & \triangleq  \lc o \in O \mathbin{\big\backslash}
 \lb\bigcup_{\theta' \gt \theta} O_{\theta'}\rb \colon \ow_\beta(o) \geq \beta \mbox{ for all } \beta \in [0, \theta] \rc. \label{eq:definition-Oi}
    \end{align}
    In particular, the auxiliary weight of any $ o \in O_\theta$ satisfies $\ow_{\theta'}(o) \geq \theta'$ for all $\theta' \leq \theta$. 
    For simplicity, we define $O_{\geq \theta} \triangleq \bigcup_{\theta' \geq \theta} O_{\theta'}$.
    \begin{claim}\label{claim:marg-O_i}
        The sets $(O_\theta)_{\theta \in [0, W]}$ partition of $O$. For any $\theta \in [0, W]$ and any $o \in O_\theta$, we have that $\ow_\theta(o) = \theta$.
    \end{claim}
    \begin{proof}
            [Proof of \Cref{claim:marg-O_i}]
            The partitioning follows by the recursive definition of these sets.
            Suppose for contradiction that $\ow_{\theta}(o) > \theta$ so in particular $\ow_\theta(o) = \theta + \e$ for some $\e >0$. We have that $\bar{w}_\theta(o) \leq \bar{w}_{\theta + \e'}(o)$ for all $\e' \in [0, \e]$ by \Cref{claim:non-increasing-auxiliary-weight}. Then since $\ow_\beta(o) \geq \beta$ for all $\beta \in [0, \theta]$, we then have that $o \in O_{\geq \theta + \e}$, contradicting the definition of $O_\theta$.
        \end{proof}

    \subsection{Properties of Incremental Values}
     We first show using submodularity that the $\prec$-incremental value satisfies certain nice behaviors, which will help analyze \Cref{alg:merge-LS}.
    \begin{mylemma}
        \label{lem:properties-incremental}
        Given an ordering $\prec$ of $E$, the following properties hold:
        \begin{compactenum}
            \item For any $S \subseteq E$, we have $f(S) - f(\emptyset) = \sum_{e \in S} w_{\prec}(e, S)$. \label{item:incremental-telescoping}
            \item For any $T \subseteq S \subseteq E$ and $e \in E$, we have $w_\prec(e, S) \leq w_\prec(e, T)$. \label{item:incremental-submodularity}
        \end{compactenum}
    \end{mylemma}
    \begin{proof}
        [Proof of \Cref{lem:properties-incremental}]
        Let $S = \{e_1, e_2, \ldots, e_{|S|}\}$ be the ordering of $S$ according to $\prec$. Then, we have that:
        \begin{align*}
            \sum_{e \in S} w_\prec(e, S) & = \sum_{j = 1}^{|S|} w_\prec(e_j, S)
            = \sum_{j = 1}^{|S|} f(e_j \mid \{e_1, e_2, \ldots, e_{j-1}\}) 
            = f(S) - f(\emptyset).
        \end{align*}
        We now prove \cref{item:incremental-submodularity}. Let $T \subseteq S$, and let $e \in E$. Let $S_e^{-}$ and $T_e^{-}$ be the set of elements appearing before $e$ in $S$ and, $T$, respectively. Using the submodular property of $f$ with the fact that $T_{e}^- \subseteq S_{e}^- $, we obtain
        \begin{align*}
            w_\prec(e, S) & = f(e \mid S_e^{-})
            \leq f(e \mid T_e^{-})
            = w_\prec(e, T). \qedhere
        \end{align*}
    \end{proof}
    The following statement is then a simple consequence of \Cref{lem:properties-incremental}
        \begin{mylemma}
            \label{claim:increasing-incremental-values}
            Fix an interval $[\beta, \lambda] \subseteq [0, W]$ and suppose that $a \in A_\tau$ contiguously for all $\tau \in [\beta, \lambda]$, i.e. $a$ is not dropped by any $(\theta, \alpha)$-swap for any $\tau \in [\beta, \lambda]$. Then, $a \in A_\tau$ for all $\tau \in [\beta, \lambda]$ and $w_{\prec_\beta}(a, A_\beta) \geq w_{\prec_\lambda}(a, A_\lambda) \geq \lambda \geq \beta$. 
        \end{mylemma}
    \begin{proof}[Proof of \Cref{claim:increasing-incremental-values}]
        The fact that $a \in A_\tau$ for all $\tau \in [\beta, \lambda]$ is by definition of $a$'s contiguous appearance.
        The second part of the lemma follows from the submodularity of the incremental value.  
        Let $(S, N)$ be the last $(\theta, \alpha)$-improvement that adds $a$ to $A_\theta$ for some $\theta \geq \lambda$.
        Let $S = \{s_1, \ldots, s_q\}$ be ordered according to the $(S,N)$ swap, and let $a = s_p$ for some $p \in [q]$. Let $X_a \subseteq X_\theta$ be the set of all elements taken by the algorithm just prior to making the $(S, N)$ swap.
        By definition of $(\theta, \alpha)$-improvement, we have that $f(a \mid X_a - \{s_p\} + \{s_1, \ldots, s_{p-1}\}) \geq \theta \geq \lambda$.
        Now, let $A'$ be the solution just prior to the addition of $a$ via the $(S, N)$-swap and $\prec$ be the ordering on $X_a$ just prior to inserting $a$ via the $(S, N)$-swap.
        Then $A' \subseteq X_a - \{s_p\} \cup \{s_1, \ldots, s_{p-1}\}$. Thus, by submodularity, we have that under $\prec$, $w_\prec(a, A') = f(a \mid A') \geq f(a \mid X_a - \{s_p\} + \{s_1, \ldots, s_{p-1}\} ) \geq \lambda$. 
        Moreover, because $a$ does not participate in other swaps within the time frame $[\beta, \theta]$, the relative ordering of $a$ diminishes as edges are (re)inserted after $a$ during swaps.
        Thus, $\{a' \in A': a' \prec a\} \supseteq \{a' \in A_\theta: a' \prec_\theta a\} \supseteq \{a' \in A_\lambda: a' \prec_\lambda a\} \supseteq \{a' \in A_\beta: a' \prec_\beta a\}$. The lemma then follows by submodularity (item~\ref{item:incremental-submodularity} of \Cref{lem:properties-incremental}).
        \end{proof}

    \section{Building Matroid Exchanges}
\label{sec: matroid swaps}
    We analyze the quality of the algorithm's solution by constructing non-improving local swaps between the intermediary solution $A_\tau$ for any $\tau \in [\btheta, W]$ and the optimal solution $O$. The information carried by the local swaps is represented by a conflict graph between $O$ and $A_\tau$. This graph has the following properties.
    \begin{restatable}{mylemma}{ConflictGraph}
        \label{thm:decomposition-ith-round-2}
        For every $\tau \in [\btheta, W]$, there exists a bipartite conflict graph $G = (O \times v(A_\tau), E)$ such that the following hold:
        \begin{compactenum}
            \item For each $v \in v(A_\tau)$, we have $\deg_G(v) = 1$. \label{item:decomp-2-1}
            \item For any $o \in O$, we have $v(A_\tau) \setminus N_o \cup v(o) \in \cI$, where $N_o \subseteq v(A_\tau)$ are the neighbors of $o$ in $G$.\label{item:decomp-swap-2}
            \item For every $o \in O_{\geq \tau}$, we have that $\deg_G(o) \geq 1$. \label{item:decomp-deg-1}
            \item For each $a \in A_\tau$, let $T_a \triangleq \{o \in O_{\geq \tau} \colon |N_o| = 1 \mbox{ and } N_o \subseteq v(a) \}$, then $\sum_{o \in T_a}\ow_\tau(o) \leq f(T_a \mid X_\tau - T_a) \leq \alpha w_{\prec_\tau}(a, A_\tau)$ and $|T_a| \leq k$.\label{item:decomp-single-2}
        \end{compactenum}
    \end{restatable}
    \Cref{thm:decomposition-ith-round-2} partitions the optimal solution $O$ and $v(A_\tau)$ into a collection of valid matroid swaps (\cref{item:decomp-swap-2}).
    It also gives further properties for the swaps involving edges of $O_{\geq \tau}$ using the absence of $(\tau, \alpha)$-improving swaps with respect to $A_\tau$. In particular, it bounds the auxiliary weight of edges $o \in O_{\geq \tau}$ incident to a fixed edge of $A_\tau$ (\cref{item:decomp-single-2}).
    These swaps will be defined by repeatedly applying matroid basis exchanges (\Cref{thm:Rota-setminus}, \Cref{thm:Rota-non-basis-2} in the appendix). The proof of this theorem follows that of \cite{Singer-Thiery:2025:Better} but adapts it to find larger exchanges up to size $k$ instead of size $2$. The proof is in \Cref{sec:appendix-matroid}.
    
\section{Approximation Guarantee of \Cref{alg:merge-LS}}\label{sec: approx analysis submod}
    We are now ready to compute the approximation guarantee of our algorithm restated here.
    \begin{mythmnonumber}
        There is a $\frac{k}{2} + o(k)$-approximation for maximizing a monotone submodular function subject to a matroid $k$-parity constraint.
    \end{mythmnonumber}
    The proof follows a standard strategy by bounding the marginal contribution of edges of $O$ with respect to their conflicting neighbors in $A$, and aggregating the bounds together. 
     Observe that $X_0 = X_{\btheta}$ are the edges that were ever inserted by a swap in \Cref{alg:merge-LS}. The analysis will consider separately the edges of $X_0\setminus A$ that appear throughout the algorithm but get removed in favor of improving swaps, and the remaining edges $A$.

\subsection{Decomposition of the value of $O$ into cardinality bounds}
    In this section, we express the value of $O$ by decomposing it according to the number of edges with the same auxiliary weight for each weight $\tau \in [0, W]$.
    To do this, for any $o \in O$, let $\tau \in [0, W]$ be the unique value such that $o \in O_\tau$ (\Cref{claim:marg-O_i}). In particular, we have $\ow_\tau(o) = \tau$ for any $o \in O$ by \Cref{claim:marg-O_i}. Then, $\ow_\tau(o)  = \int_{0}^\tau 1 dy$.
    \begin{mylemma}
        \label{lem:decomp-opt-value}
        The following equation holds:
        \begin{align}
            \label{eq:opt-value-decomp}
            f(O \mid X_0) & \leq \delta f(A) + \int_{\btheta}^W |O_{\geq \tau}| d\tau\enspace.
        \end{align}
    \end{mylemma}

    \begin{proof}[Proof of \Cref{lem:decomp-opt-value}]
        By definition of the auxiliary weights $\ow$, we have that 
        \begin{align*}
            f(O \mid X_0) & \leq \displaystyle\sum_{o \in O}\ow_0(o) \leq \int_{0}^W \displaystyle\sum_{o \in O_\tau} \ow_\tau(o) d\tau
        \end{align*}
        where we have used that $\ow_\tau(o)$ is increasing in $\tau$ and that $(O_\tau)_{\tau \in [0, W]}$ partition $O$. Then using that $\ow_\tau(o) = \tau$ for $o \in O_\tau$ (\Cref{claim:marg-O_i}), we get that 
        \begin{align*}
            f(O \mid X_0) &\leq  \int_{0}^W \displaystyle |O_\tau| \tau d\tau
            = \int_{0}^W \int_0^\tau |O_\tau| dy d\tau
            = \int_0^W \int_y^W |O_\tau| d\tau dy 
            = \int_0^W |O_{\geq y}| dy, 
        \end{align*}
        Finally, we split this last integral depending on the value of $\tau$ to get
        \begin{align*}
            f(O \mid X_0) & \leq \int_0^{\btheta} |O_{\geq y}| dy + \int_{\btheta}^W |O_{\geq y}| dy \leq |O|\btheta + \int_{\btheta}^W |O_{\geq y}| dy \leq \delta W + \int_{\btheta}^W |O_{\geq y}| dy,
        \end{align*}
        The last inequality uses that $\btheta = \delta \frac{W}{|E|}$. The proof concludes using that $W \leq f(A)$.
    \end{proof}

    We then use the bipartite conflict-graph from \Cref{thm:decomposition-ith-round-2} to bound the size of $|O_{\geq \tau}|$ with respect to $|A_\tau|$ for any $\tau \in [\btheta, W]$.
    \begin{mylemma}
        \label{lem:size-O_i}
        For any $\tau \in [\btheta, W]$, we have
        \begin{align*}
            \card{O_{\geq \tau}} & \leq \frac{k}{2} \cdot |A_\tau| + \sum_{a \in A_\tau} \min \lc k; \frac{\alpha w_{\prec_\tau}(a, A_\tau)}{\tau} \rc\enspace.
        \end{align*}
    \end{mylemma}
    This statement follows a simple counting argument. Each $o \in O_{\geq \tau}$ has degree at least $1$ in the bipartite graph $G$ given in \Cref{thm:decomposition-ith-round-2} (i.e. at least one matroid-conflict). We obtain the result by splitting $O_{\geq \tau}$ into vertices with degree $1$ or more than $2$, using that every $a \in A_\tau$ has degree at most $k$, and that there are not too many vertices of degree $1$.
    \begin{proof}
        [Proof of \Cref{lem:size-O_i}]
        The proof follows from the constructed bipartite conflict graph on $O \times v(A_\tau)$ of \Cref{thm:decomposition-ith-round-2}. Partition $O_{\geq \tau} = O_s \sqcup O_d$, where $O_s = \{o \in O_{\geq \tau} \colon |N(o)| = 1\}$, and $O_d = \{o \in O_{\geq \tau} \colon |N(o)| \geq 2\}$, since every $o \in O_{\geq \tau}$ satisfies $\deg_G(o) \geq 1$ (\cref{item:decomp-deg-1}~\Cref{thm:decomposition-ith-round-2}).
        These sets denote the edges of $O$ with a \emph{single} conflict and the edges with at least 2 conflicts with respect to $v(A_\tau)$.
        We know that every vertex $v \in v(A_\tau)$ is contained in at most 1 neighborhood $\{N_o\}_{o \in O}$. In particular, we have $2 |O_d| \leq \sum_{o \in O_d} |N(o)| \leq |v(A_\tau)| \leq k |A_\tau|$. \\

        Moreover, the sets $\{T_a\}_{a \in A_\tau}$ from item \ref{item:decomp-single-2} of \Cref{thm:decomposition-ith-round-2} partition $O_s$ and satisfy $|T_a| \leq k$ and $\displaystyle\sum_{o \in T_a} \ow_\tau(o) \leq \alpha w_{\prec_\tau}(a, A_\tau)$. Lastly, since $o \in O_{\geq \tau}$, we have that $\ow_\tau (o) \geq \tau$ by definition. Then,
        \begin{align*}
            |O_s| = \displaystyle\sum_{a \in A_\tau} |T_a| \leq \displaystyle\sum_{a \in A_\tau} \min \lc k ; \frac{\alpha w_{\prec_\tau}(a, A_\tau)}{\tau} \rc.
        \end{align*}
        We conclude the lemma since $|O_{\geq \tau}| = |O_s| + |O_d|  \leq \frac{k}{2} \cdot |A_\tau|+ \sum_{a \in A_\tau} \min \lc k; \frac{\alpha w_{\prec_\tau}(a, A_\tau)}{\tau} \rc. \qedhere$    \end{proof}

\subsection{Contribution of Appearances}
    We would now like to integrate the inequality of \Cref{lem:size-O_i} over $\tau \in [\btheta, W]$ to obtain a guarantee in terms of the value of the final solution. We will split the analysis into the edges of $A_\tau$ which appear in the final solution $A$ and all occurrences of edges that get removed throughout the process. 
    For this purpose, we need to define \emph{appearances} for every edge in $X_0$.
    Recall that $X_0$ is the set of all edges that ever entered the solution throughout the execution of \Cref{alg:merge-LS}. 
    \begin{mydef}
        \label{def:appearance}
        For every edge $a \in X_0$, we assign a set of \emph{appearances}.
        An appearance of $a$ is a maximal interval $(\theta_r, \theta_i]$ in $[\btheta, W]$ such that \begin{enumerate}
            \item $a$ is inserted during $(\theta_i, \alpha)$-$\LS$ 
            \item $a$ is \textbf{not} removed during $(\theta, \alpha)$-$\LS$ for any $\theta \in (\theta_r, \theta_i]$ if $\theta_r \neq \theta_i$
            \item either $(i)$ $a$ is removed during $(\theta_r, \alpha)$-$\LS$ (an \emph{evicted appearance}) or \\
           \-\hspace{1cm} $(ii)$ $a$ remains in the solution until $\theta_r = 0$ (a unique \emph{final appearance}).
        \end{enumerate}
        Note that we can have $\theta_i = \theta_r$ if $a$ is inserted and removed during $(\theta_i, \alpha)$-$\LS$.
    \end{mydef}
    Informally, an appearance is a continuous interval during which an edge stays within the solution. An edge may have multiple appearances.
    \begin{mydef}
        \label{def:appearance-notation}
        For an edge $a \in X_0$, we denote by $\Kcal_a$ the set of its evicted appearances. We denote by $\Fcal_a$ the final appearance, which is empty for $a \notin A$ and nonempty for $a \in A$.
    \end{mydef}

Note that every evicted appearance of some $a \in X_0$ must be part of some $(\theta, \alpha)$-swap $(S, N)$ where $a \in N$ removes $a$ from the solution. We will aggregate over all such swaps to bound the contribution of the evictions. Towards this end, let $T$ be the total number of $(\theta, \alpha)$-swaps performed by \Cref{alg:merge-LS} (over all $\theta$). Let $A^{(1)} = \emptyset$ be the initial solution and $A^{(t)}$ be the solution just prior to performing the $\nth{t}$ exchange for $t \in [T]$, where $N_t$ is discarded and $S_t$ is added. Let $\prec_{t}$ be the order over edges $A^{(t)}$ in the solution prior to the $\nth{t}$ swap. Let $A^{(T+1)} = A$ be the final output of \Cref{alg:merge-LS} after the $\nth{T}$ swap. 

\begin{mydef}\label{def: tth swap}
    For any $a \in X_0$ and any evicted appearance $I = (\theta_r, \theta_i] \in \Kcal_a$, there exists a swap $(S_t, N_t)$ made during $(\theta_r, \alpha)$-$\LS$ where $a \in N_t$ removes $a$. We denote by $t_I$ the index $t \in [T]$ associated with this swap $(S_t, N_t)$.  
\end{mydef}

\begin{claim}\label{claim: evict marginal bound}
    For any $a \in X_0$ and any evicted appearance $I = (\theta_r, \theta_i] \in \Kcal_a$, we have $\gamma \leq w_{\prec_{\gamma}}(a, A_{\gamma}) \leq w_{\prec_{t_I}}(a, A^{(t_I)})$ for any $\gamma \in I$.
\end{claim}
\begin{proof}[Proof of \Cref{claim: evict marginal bound}]
    This claim is a direct application of \Cref{claim:increasing-incremental-values}, where we use that $a$ is inserted during a $(\theta_i, \alpha)$-swap, and does not participate in any swaps until the removal $(S_{t_I}, N_{t_I})$ during $(\theta_r, \alpha)$-$\LS$ so that $\{a' \in A_{\gamma}: a' \prec_{\gamma} a\} \supseteq \{a' \in A^{(t_I)}: a' \prec_{t_I} a\}$ for any $\gamma \in (\theta_r, \theta_i]$.
\end{proof}

\begin{notation}
    From this point on, and unless specified otherwise, we remove the subscript $w_{\prec_\theta}(a, A_\theta)$ and instead write $w(a, A_\theta)$ as we only evaluate the incremental value of $a$ with respect to the ordering defined by $A_{\theta}$. The same applies for $w_{\prec_t}(a, A^{(t)})$ which we write as $w(a, A^{(t)})$.
\end{notation}

\begin{claim}\label{claim: N S swap bound}
    For any swap $(S_t, N_t)$ for $t \in [T]$, we have that $f(A^{(t+1)}) - f(A^{(t)}) \geq (\alpha -1) \cdot \displaystyle\sum_{a \in N_t}w(a, A^{(t)})$.
\end{claim}

\begin{proof}[Proof of \Cref{claim: N S swap bound}]
    Consider a $(\theta, \alpha)$-improvement $(S_t, N_t)$ for the current solution $A^{(t)}$. Let $X^{(t)}$ the set of all edges that have appeared up until reaching $A^{(t)}$. Then, we have that
        \begin{align*}
            f(A^{(t+1)}) - f(A^{(t)}) & = f((A^{(t)} \setminus N_t) \cup S_t) - f(A^{(t)}) \\
            & = f(S_t \mid A^{(t)} \setminus (N_t \cup S_t)) - f(N_t \mid A^{(t)} \setminus N_t) & (\mbox{since } S_t \not\subseteq A^{(t)})\\
            & \geq f(S_t \mid X^{(t)} - S_t) - f(N_t \mid A^{(t)} \setminus N_t) &(\mbox{submodularity: } A^{(t)} \subseteq X^{(t)}) \\
            & \geq f(S_t \mid X^{(t)} - S_t) - \sum_{a \in N_t} w(a, A^{(t)})\\
            & \geq (\alpha - 1) \cdot \sum_{a \in N_t} w(a, A^{(t)}). & (\mbox{definition of } (\tau, \alpha)\mbox{-swap})
        \end{align*}
        For the second inequality, let $N_t = \{a_1, \ldots, a_\ell\}$ be an enumeration of $N_t$ according to $\prec$, then by submodularity we have:
        \begin{align*}
            f(N_t \mid A^{(t)} \setminus N_t) & = \sum_{i = 1}^\ell f(a_i \mid A^{(t)} \setminus N_t \cup \{a_1, \ldots, a_{i-1}\}) \leq \sum_{i = 1}^\ell w(a_i, A^{(t)})\enspace. \qedhere
        \end{align*}
\end{proof}

    The following lemma bounds the effect of \emph{evicted} appearances compared to the value of the final solution. In particular, we prove that the value of evicted edges is small compared to the weight of the final solution.
    \begin{mylemma}
        \label{lem:bound-X}
        We have that
        \begin{align*}
            \sum_{a \in X_0} \displaystyle\sum_{I \in \Kcal_a} w(a, A^{(t_I)}) & \leq \frac{1}{\alpha - 1} f(A) \enspace.
        \end{align*}
    \end{mylemma}
    \begin{proof}
        [Proof of \Cref{lem:bound-X}]
        Using \Cref{claim: N S swap bound} and the fact that $A^{(T+1)} = A$ the final algorithm output, we have that
        \begin{align*}
            \sum_{a \in X_0} \displaystyle\sum_{ I \in \Kcal_a} w(a, A^{(t_I)})
            & \leq \sum_{t \in [T]}\displaystyle\sum_{a \in N_t} w(a, A^{(t)})
            \leq \sum_{t \in [T]} \frac{1}{\alpha -1}\left[ f(A^{(t+1)}) - f(A^{(t)})\right]
            \leq \frac{f(A)}{\alpha -1}. \qedhere
        \end{align*}
    \end{proof}

    It remains to show that we can integrate the right hand side of \Cref{lem:size-O_i} and express the value as a function of evicted appearances.
    \begin{mylemma}
        \label{lem:doubles-occurences}
        The following equation holds:
        \begin{align*}
            \int_{\btheta}^W  |A_\tau|d\tau & \leq  \frac{\alpha}{\alpha - 1} f(A).
        \end{align*}
    \end{mylemma}
    \begin{proof}
        [Proof of \Cref{lem:doubles-occurences}]
        The core of the proof lies in showing that:
        \begin{align*}
            \int_{\btheta}^W  |A_\tau|d\tau & \leq \left(\sum_{a \in X_0}\displaystyle\sum_{I \in \Kcal_a} w(a, A^{(t_I)}) \right)+ f(A),
        \end{align*}
        to then apply \Cref{lem:bound-X}.
        Observe that for any $\tau \in [\btheta, W]$ and each $a \in A_\tau$, we can associate a unique appearance $ (\theta_r, \theta_i] \in \Kcal_a \cup \Fcal_a$ such that $\tau \in (\theta_r, \theta_i]$ is covered by the appearance. Therefore,
        \begin{align*}
            \int_{\btheta}^W |A_\tau| d\tau  \leq \displaystyle\sum_{a \in X_0} \displaystyle\sum_{(\theta_r, \theta_i] \in \Kcal_a \cup \Fcal_a} \int_{t_r}^{t_i} 1\, d\tau
            & = \displaystyle\sum_{a \in X_0} \displaystyle\sum_{(\theta_r, \theta_i] \in \Kcal_a \cup \Fcal_a} (\theta_i - \theta_r) \\
            &\leq \displaystyle\sum_{a \in X_0} \displaystyle\sum_{I \in \Fcal_a } w(a, A_{\theta_i}) + \displaystyle\sum_{a \in X_0} \displaystyle\sum_{I \in \Kcal_a} w(a, A^{(t_I)}) \\
            & \leq \displaystyle\sum_{a \in X_0} \displaystyle\sum_{I \in \Fcal_a } w(a, A_{0}) + \displaystyle\sum_{a \in X_0} \displaystyle\sum_{I \in \Kcal_a} w(a, A^{(t_I)})
        \end{align*}
        where the last two inequalities follow from $\theta_i \leq w(a, A_{\theta_i}) \leq w(a, A_{0})$ for any final appearance by \Cref{claim:increasing-incremental-values}. We also use that $\theta_i \leq w(a, A_{\theta_i}) \leq w(a, A^{(t_I)}$ for any evicted appearance (\Cref{claim: evict marginal bound}).
        Since $A_0 = A$ and there is only a final appearance if $a \in A$, we get
        \begin{align*}
             \int_{\btheta}^W |A_\tau| d\tau & \leq \displaystyle\sum_{a \in X_0} \displaystyle\sum_{I \in \Kcal_a} w(a, A^{(t_I)}) + \displaystyle\sum_{a \in A} w(a, A_0)
             \leq \frac{1}{\alpha -1} f(A) + f(A)
        \end{align*}
        where the last inequality follows from applying \Cref{lem:bound-X} to the first summand and the telescoping property of the incremental value to the second summand. This yields the desired result.
    \end{proof}

    \subsection{Contribution of Singles}
    Recall that the second term on the right-hand side of \Cref{lem:size-O_i} comes from summing over the single conflict edges $O_s$. We proceed similarly to bound the total value by integrating this term over all $\tau$ and we show that this term contributes an $o(k)$ value of $f(A)$.
    \begin{mylemma}
        \label{lem:single-occurences}
        The following equation holds:
        \begin{align*}
            \int_{\btheta}^W  \sum_{a \in A_\tau}\min \lc k, \frac{\alpha w_{\prec_\tau}(a, A_\tau)}{\tau}\rc d\tau \leq \lb \alpha + \alpha \ln(k/\alpha) \rb \frac{\alpha}{\alpha - 1} f(A).
        \end{align*}
    \end{mylemma}
    \begin{proof}
        [Proof of \Cref{lem:single-occurences}]
        As in the proof of \Cref{lem:doubles-occurences}, we note that each $a \in A_\tau$ is associated to a unique appearance $(\theta_r, \theta_i] \in \Kcal_a \cup \Fcal_a$ such that $\tau \in (\theta_r ,\theta_i]$. Therefore,
        \begin{align*}
            \int_{\btheta}^W  \sum_{a \in A_\tau}\min \lc k, \frac{\alpha w(a, A_\tau)}{\tau}\rc d\tau
            & \leq \sum_{a \in X_0} \sum_{(\theta_r, \theta_i] \in \Kcal_a \cup \Fcal_a}\int_{\theta_r}^{\theta_i} \min \lc k, \frac{\alpha w(a, A_\tau)}{\tau}\rc d\tau.
        \end{align*}
        Consider a fixed interval $(\theta_r, \theta_i] \in \Kcal_a \cup \Fcal_a$. Notice that for any $\tau \in (\theta_r, \theta_i]$, the term $\frac{\alpha w(a, A_\tau)}{\tau}$ is decreasing in $\tau$ by \Cref{claim:increasing-incremental-values} and since $\frac{1}{\tau}$ decreases in $\tau$. 
        Let $\gamma \in (\theta_r, \theta_i]$ be such that $k \leq \frac{\alpha w(a, A_q)}{q}$ for all $q \in (\theta_r, \gamma]$, and  $k \geq \frac{\alpha w(a, A_q)}{q}$ for all $q \in (\gamma, \theta_i]$. 
        Then, for a fixed interval $(\theta_r, \theta_i]$, we have
        \begin{align*}
            \int_{\theta_r}^{\theta_i} \min \lc k, \frac{\alpha w(a, A_\tau)}{\tau}\rc d\tau & = \int_{\theta_r}^\gamma k d\tau + \int_{\gamma}^{\theta_i} \frac{\alpha w(a, A_\tau)}{\tau} d\tau. 
            \end{align*}
        We split the computation into two parts and estimate each integral separately.
        The first integral is vacuous if $\gamma = \theta_r$, otherwise, we have $k \leq \frac{\alpha w(a, A_{\gamma})}{\gamma}$. Thus, we have that
        \begin{align*}
            \int_{\theta_r}^{\gamma} k d\tau & = (\gamma - \theta_r) k \leq \gamma k \leq \alpha w(a, A_\gamma).
            %\leq \alpha w(a, A_{t_r}),
        \end{align*}
          The second integral is vacuous if $\gamma = \theta_i$, otherwise, we have that $\gamma \geq \frac{\alpha w(a, A_{\gamma})}{k}$. Therefore, we get that
          \begin{align*}
              \int_{\gamma}^{\theta_i} \frac{\alpha w(a, A_\tau)}{\tau} d\tau & \leq \int_{\gamma}^{w(a, A_\gamma)} \frac{\alpha w(a, A_\gamma)}{\tau} d\tau \\
            & = \alpha w(a, A_\gamma) \ld \ln(w(a, A_\gamma)) - \ln(\gamma) \rd \\
            & \leq \alpha w(a, A_\gamma) \cdot \ld \ln(w(a, A_\gamma)) - \ln(\alpha w(a, A_\gamma)/k)) \rd\\
            & = \alpha w(a, A_{\gamma}) \cdot \ln(k/\alpha). 
          \end{align*}
        The first inequality uses \Cref{claim:increasing-incremental-values} such that $w_{\prec_\tau}(a, A_\tau) \leq w_{\prec_\gamma}(a, A_{\gamma})$ for $\tau \in (\gamma, \theta_i]$ and $\theta_i \leq w_{\prec_\gamma}(a, A_\gamma)$.
        Combining these two equations, we get that 
        \begin{align}
            \int_{\theta_r}^{\theta_i} \min \lc k, \frac{\alpha w(a, A_\tau)}{\tau}\rc d\tau
            & \leq \alpha w(a, A_\gamma) \lb 1 + \log(k/\alpha)\rb.\label{eq:interval integral}
        \end{align}
        Notice that if $I = (\theta_r, \theta_i]$ is in $ \Kcal_a$, then $w_{\prec_\gamma}(a, A_\gamma) \leq w_{\prec_{t_I}}(a, A^{(t_I)})$ by \Cref{claim: evict marginal bound}. Summing \Cref{eq:interval integral} over all $a \in X_0$ and $I \in \Kcal_a$ and using \Cref{lem:bound-X}, we then get that 
        $$\sum_{a \in X_0} \sum_{(\theta_r, \theta_i] \in \Kcal_a}\int_{\theta_r}^{\theta_i} \min \lc k, \frac{\alpha w(a, A_\tau)}{\tau}\rc d\tau \leq  \lb 1 + \ln(k/\alpha) \rb \frac{\alpha f(A)}{\alpha -1}.$$
        Moreover, summing \Cref{eq:interval integral} over all $a \in A$ and $I \in \Fcal_a$ and using that $w_{\prec_\gamma}(a, A_\gamma) \leq w_{\prec_0}(a, A_0)$ for $(\theta_r, \theta_i] \in \Fcal_a$ and $\gamma \in (\theta_r, \theta_i]$, we get that 
        $$\sum_{a \in X_0} \sum_{I \in \Fcal_a}\int_{\theta_r}^{\theta_i} \min \lc k, \frac{\alpha w(a, A_\tau)}{\tau}\rc d\tau \leq  \lb 1 + \ln(k/\alpha) \rb \alpha f(A).$$
        We conclude the proof by putting the two bounds together to obtain that
        \begin{align*}
            \int_{\btheta}^W \sum_{a \in A_\tau} \min \lc k, \frac{\alpha w(a, A_\tau)}{\tau}\rc d\tau 
            & \leq \lb \alpha + \alpha \ln(k/\alpha) \rb \frac{\alpha}{\alpha - 1} f(A). \qedhere
        \end{align*}
    \end{proof}

\subsection{Putting Guarantees Together to Prove \Cref{thm:main-approximation-guarantee}}
    As a corollary of the previous computations, we directly obtain a bound on the marginal contribution of $O$ with respect to $X_0$ by combining \Cref{lem:decomp-opt-value} together with \Cref{lem:size-O_i}, \Cref{lem:doubles-occurences}, and \Cref{lem:single-occurences}.
    \begin{mycor}
        \label{cor:marg-OPT-wrt-X0}
        The following equation holds:
        \begin{align*}
            f(O \mid X_0) & \leq \frac{\alpha}{\alpha - 1} \lb \frac{k}{2} +  \alpha + \alpha \ln(k/\alpha) + \delta \rb f(A).
        \end{align*}
    \end{mycor}
    \begin{proof}[Proof of \Cref{cor:marg-OPT-wrt-X0}]
        By \Cref{lem:decomp-opt-value}, \Cref{lem:size-O_i}, \Cref{lem:doubles-occurences}, and \Cref{lem:single-occurences}, we have
        \begin{align*}
            f(O \mid X_0) - \delta f(A) & \leq \int_{\btheta}^W |O_{\geq \tau}| d\tau & (\mbox{\Cref{lem:decomp-opt-value}})\\
            & \leq \int_{\btheta}^W \frac{k}{2} \cdot |A_\tau| + \sum_{a \in A_\tau} \min \lc k; \frac{\alpha w(a, A_\tau)}{\tau} \rc d\tau & (\mbox{\Cref{lem:size-O_i}})\\
            &\leq \frac{k}{2} \cdot \frac{\alpha}{\alpha -1}f(A) + \lb \alpha + \alpha \ln(k/\alpha)\rb \frac{\alpha}{\alpha -1}f(A) & (\mbox{\Cref{lem:doubles-occurences},\ref{lem:single-occurences}})\\
            & = \frac{\alpha}{\alpha - 1} \lb \frac{k}{2} +  \alpha + \alpha \ln(k/\alpha)\rb f(A). \qedhere
        \end{align*}
    \end{proof}

    Observe that \Cref{cor:marg-OPT-wrt-X0} only bounds $f(O \mid X_0)$ and not $f(O \mid A)$. The next lemma shows that $f(X_0)$ is close to $f(A)$.
    \begin{mylemma}
        \label{lem:value-inserted-wrt-final}
        Let $A$ be the final solution, and let $X_0$ be the set of all edges entering the solution throughout the execution of \Cref{alg:merge-LS}. Then, the following equation holds:
        \begin{align*}
            f(X_0) & \leq \frac{\alpha}{\alpha - 1} f(A) \enspace.
        \end{align*}
    \end{mylemma}
    \begin{proof}
        [Proof of \Cref{lem:value-inserted-wrt-final}]
Define another ordering $\prec_\ast$ over $X_0$ where an edge is inserted in $\prec_{\ast}$ when it is first inserted by \Cref{alg:merge-LS} and never replaced after that. Contrary to $\prec_0$, the ordering $\prec_\ast$ orders edges according to their \textit{first} insertion by the algorithm only. 
Any edge $a \in X_0 \setminus A$ was evicted by some appearance $I \in \Kcal_a$. For each such $a$, let $I_a$ be the first such evicted appearance of $a$. Then, 
\begin{align*}
    f(X_0) &= \displaystyle\sum_{a \in X_0} f(a \mid \{a' \in X_0: a' \prec_\ast a\}) \\
    & = \displaystyle\sum_{a \in A} f(a \mid \{a' \in X_0: a' \prec_\ast a\}) + \displaystyle\sum_{a \in X_0\setminus A} f(a \mid \{a' \in X_0: a' \prec_\ast a\}) \\
    & \leq \displaystyle\sum_{a \in A} f(a \mid \{a' \in A: a' \prec_\ast a\}) + \displaystyle\sum_{a \in X_0\setminus A} f(a \mid \{a' \in A^{(t_{I_a})}: a' \prec_\ast a\}) \\
    & \leq f(A) + \displaystyle\sum_{a \in X_0\setminus A} w_{\prec_{t_{I_a}}}(a, A^{(t_{I_a})}) \\
    & \leq \left( 1 + \frac{1}{\alpha - 1} \right) f(A)
\end{align*}
where the penultimate inequality follows by submodularity and from the fact that $\{a' \in A^{(t_{I^a})}: a' \prec_{t_{I^a}} a\} \subseteq \{a' \in A^{(t_{I^a})} : a' \prec_\ast a\}$ for any $a \in X_0 \setminus A$. The last inequality follows from \Cref{lem:bound-X}.\qedhere
    \end{proof}

    We can now state the guarantees of \Cref{alg:merge-LS} precisely.
    \begin{mythm}
        \label{thm:main-approx}
        Let $A$ be the output of \Cref{alg:merge-LS}, and let $O$ be the optimal solution. Then, the approximation guarantee of \Cref{alg:merge-LS} is equal to
        \begin{align*}
            f(O) & \leq \frac{\alpha}{\alpha - 1} \lb \frac{k}{2} + 1 + \alpha + \alpha \log(k/\alpha) + \delta \rb f(A).
        \end{align*}
    \end{mythm}
    \begin{proof}[Proof of \Cref{thm:main-approx}]
        We combine \Cref{cor:marg-OPT-wrt-X0}, and \Cref{lem:value-inserted-wrt-final} and monotonicity of $f$ to obtain
        \begin{align}
            f(O) \leq f(O \mid X_0) + f(X_0) & \leq \frac{\alpha}{\alpha - 1} \lb \frac{k}{2} +  \alpha + \alpha \log(k/\alpha) + \delta\rb f(A) + \frac{\alpha}{\alpha-1} f(A). \qedhere
        \end{align}
    \end{proof}

    Setting $\alpha$ appropriately, we obtain the claimed asymptotic guarantee.
       \begin{restatable}{mycor}{alphacalc}
        \label{cor:approx-num}
        Set $\alpha = \sqrt{k/\log(k)}$ in \Cref{alg:merge-LS}. Then,
        \begin{align*}
            f(O) & \leq \lb \frac{k}{2} + O(\sqrt{k \log(k)}) \rb f(A).
        \end{align*}
    \end{restatable}
    
   The calculations that conclude \Cref{cor:approx-num} are given in Appendix \ref{sec: appendix-alpha-calc}.

    \section{An improved approximation for set packing via hybrid $\alpha$-local search and unweighted local search}\label{sec: set pack}

The weighted $k$-set packing problem is a subproblem of weighted matroid $k$-parity, in which the underlying matroid is free, i.e., $\cM = 2^V$. In particular, the goal is simply to find a maximum weight collection of disjoint edges, each of which contains up to $k$ vertices. We denote by $w:E \rightarrow \mathbb{R}_{\geq 0}$ the weight function.
For any two sets of edges $A, O \subseteq E$, we define the \emph{neighborhood} of $O$ in $A$ as $N(O, A) \triangleq \{ a \in A \colon a \cap O \neq \emptyset\}$ as the set of edges in $A$ with common vertices with $O$. The set of vertices covered by both hyperedges of $A$ and $O$ in is denoted by $N_v(O,A) \triangleq \{v(O) \cap v(A)\}$.
    We also define a weight function over the vertices of the set packing instance.
    \begin{mydef}[Vertex Weight]
        Let $A \subseteq E$ be a collection of disjoint hyperedges, we define $c:v(A) \rightarrow \RR_{\geq 0}$ where $c_v = w_e$ for $v \in e$, and $c(S) = \sum_{v \in S} c_v$ for all $S \subseteq v(A)$. Thus, $c_v$ is the cost of the hyperedge $e\in A$ covering $v$. This is well defined since $A$ is a collection of disjoint hyperedges.
    \end{mydef}

\subsection{Main Algorithm}

We obtain a  $\frac{\ln(4) k}{3} + o(k)$-approximation algorithm for weighted $k$-Set Packing by combining the $\alpha$-local search algorithm and the algorithm of Singer and Thiery \cite{Singer-Thiery:2025:Better}. 
The algorithm of \cite{Singer-Thiery:2025:Better} creates a random partitioning of the instance into weight classes defined by a random parameter $\tau$. The algorithm processes the weight classes greedily and runs the state-of-the-art unweighted algorithm on each weight class. Further details can be found in \cite{Singer-Thiery:2025:Better, Feldman-Ward:2026:Submodular}.
Our algorithm repeats this strategy, but after processing the $\nth{i}$ weight class, it continues to search for improving $\alpha$-swaps over all the edges up to the $\nth{i}$ weight class. In particular, in contrast to the method outlined in \cite{Singer-Thiery:2025:Better, Feldman-Ward:2026:Submodular}, we do not contract the instance at each iteration but rather continue to search for $\alpha$ improvements over the increasingly larger instance. 
It is worth noting that to apply the state-of-the-art unweighted algorithm \cite{Cygan:2013:Improved}, we treat each weight class as completely unweighted (\Cref{alg:interval-local-search}).

    \paragraph{Interval Local Search:}
        Given a feasible solution $A$, \Cref{alg:interval-local-search} applies Cygan's algorithm \cite{Cygan:2013:Improved} on an interval $I$. It does so by defining a subinstance over all edges with $w_e \in I$ which are unblocked by the current solution $A$. The algorithm of \cite{Cygan:2013:Improved} runs with a precision parameter $\e$ which also affects the run time. We will choose $\e$ to be an arbitrarily small constant.
    \begin{algorithm}[h!]
        \DontPrintSemicolon
        \caption{\textsc{Interval} \LS$(A, I, \e)$}
        \label{alg:interval-local-search}
        \KwIn{Feasible solution $A \subseteq E$, and interval $I \subseteq [0, W]$.}
        Let $E' \leftarrow \{e \in E \colon w_e \in I \mbox{ and } N(e, A) = \emptyset \}$\;
        Let $A' \leftarrow \textsc{Unweighted Local-Search}(E', \e)$\;
        \Return $A \cup A'$\;
    \end{algorithm}

    \paragraph{Hybrid Local Search:} \Cref{alg:merge-LS-setpacking} is our main algorithm which uses \Cref{alg:interval-local-search} and \Cref{alg:local-search} as subroutines.
    For clarity, we restate the definition of $(\theta, \alpha)$-improving swaps for the weighted setting where these swaps are modified for to run over elements of weight strictly greater than $\theta$.
    
        \begin{mydef}[$(\theta, \alpha)$-improvement]
        \label{def:improvement-weight}
            Let $A \subseteq E$ be a feasible solution.
            Given a constant $\theta > 0$ and $\alpha \geq 1$, we say that a pair of sets $(S, N)$ with $S \subseteq E \setminus A$ and $N \subseteq A$ is a \emph{$(\theta, \alpha)$-improvement for $A$} if $(A \cup S) \setminus N$ is feasible and if
            \begin{align}
            |S| = \ell \leq k, |N| \leq \ell k \text{, each }s \in S\text{ satisfies }w_s > \theta\text{, and }w(S) \geq \alpha \cdot
            w(N).
            \label{item:alpha-swap-weight}
            \end{align}
        \end{mydef}

    \begin{algorithm}[h!]
            \DontPrintSemicolon
            \SetKw{KwBy}{by}
            \caption{\HMLS$(\alpha, \e)$}
            \label{alg:merge-LS-setpacking}
            \KwIn{Parameter %$\delta$ that quantifies the number of intervals,%
            $\e$ that quantifies the precision of \textsc{Interval $\LS$}, and $\alpha$ that quantifies the swap value size.}
            % Let $1 - \tau \sim \mathcal{U}(1/2, 1]$ \hfill \tcp{Random marker placement}
            Let $\beta \sim \mathcal{U}(0,1]$ and let $\tau \leftarrow 2^\beta$\hfill\tcp{Random marker placement}
            % Let $L \triangleq \lceil -\log_{1/2}(\card{E}\delta^{-1}) \rceil + 1$ \hfill \tcp{Number of intervals}
            Let $W \gets \max_{e \in E} w_e$\;
            Let $m_i \gets W\tau2^{-i}$ for any $i \in \mathbb{N}$\;
            Let $\mathcal{J} \gets \{j \in \mathbb{N} : \exists e \in E, w_e \in (m_{j+1}, m_j]\}$ \hfill \tcp{Collection of relevant intervals}
            Let $A \leftarrow \emptyset$ \hfill \tcp{Initialize $A$}
            \For{$i \in \mathcal{J}$}
            {%\tcp{At the beginning of this loop: $A \leftarrow \textsc{Interval } \LS(A, (m_{i}, m_0])$}
            {$A \gets \textsc{Interval } \LS(A, (m_{i+1}, m_{i}], \e)$\;
            }
            $A \gets  (m_{i+1}, \alpha)$-$\LS(A)$ \hfill \tcp{\Cref{alg:local-search}}
            }
            \Return $A$
    \end{algorithm}
        To see that \Cref{alg:merge-LS-setpacking} runs in polynomial time in $|E|, \e$, we note that we only process the intervals that contain at least one element of the ground set, so that there are at most $|E|$ many intervals and therefore iterations of $\textsc{Interval } \LS$ and $(m_{i+1}, \alpha)$-$\LS$. In Appendix~\ref{sec:runtime} we argue that $(m_{i+1}, \alpha)$-$\LS$ runs in polynomial time in $|E|$ for constant $k$. $\textsc{Interval } \LS$ runs the algorithm of \cite{Cygan:2013:Improved} on precision $\e$ which runs in polynomial time in $|E|$ for any choice of constant $\e$.

\subsection{Analysis of \Cref{alg:merge-LS-setpacking}}
 Our main theorem states that \Cref{alg:merge-LS-setpacking} gives an improved approximation for $k$-Set Packing. 
 \begin{restatable}{mythm}{setpackapprox}
        \label{thm:APX-better}
        There exists a choice of parameters $\alpha \geq 1, \e \in (0, 1)$ for \Cref{alg:merge-LS-setpacking} such that its final output $A$ satisfies:
        \begin{align*}
            w(O) & \leq \left(\frac{\ln(4)k}{3}+ O(\sqrt{k})\right) \esp{w(A)} \leq \left(\frac{k}{2.16404}+ O(\sqrt{k})\right) \esp{w(A)}.
        \end{align*}
 \end{restatable}
  This improves over Neuwohner's $\frac{k}{2.00561} + O(1)$-approximation algorithm without relying on the use of the $w^2$-algorithm of Berman \cite{Berman:2000:Approximation}. \\
 
 To prove the above theorem, we define a partition of the elements of $O \cap (m_{i+1}, m_i]$ right after the termination of $A_i \leftarrow \textsc{Interval} \LS((m_{i+1}, m_i])$. It indicates whether the neighborhood of an optimal edge contains 1, 2, or more than 3 vertices in $A_i$ (i.e. single, double or triple conflicts). The proof splits the analysis into neighbors present in final output and those removed.
 The idea is that, on each interval, Cygan's analysis of $\textsc{Interval } \LS$ \cite{Cygan:2013:Improved} proves that the majority of $O$, except an $o(k)$-fraction, have $3$ vertex conflict. 
 They behave roughly as a $k/3$ approximation. 
 The neighbors that get removed account for a $\alpha$ factor weight increase of the current solution. 
 Putting all of these influences together yields the improved approximation guarantee. 
    % \begin{restatable}[Projection]{mydef}{Projection}
    %     Let $A$ be the current solution of \Cref{alg:merge-LS-setpacking} and let $O$ be an optimal solution. 
    %     For each $o \in O$, we define the \emph{projection} as $\proj(o, A) \triangleq v(N(o,A))$ as the vertices covered by the neighborhood of $o$ in $A$.
    %     %. Let $\{o, N_o\}_{o \in O}$ be the set of exchanges between $O$ and $A$, where since we are a set packing instance, $N_o$ is exactly the set of hyperedges in $A$ that are adjacent to the hyperedge $o$.  the mapping $\proj(o, A) = v(N_o)$, 
    % \end{restatable}
    % \begin{remark}
    %     Since $O$ is a collection of disjoint edges, for any solution $A$, we have $N(o, A) \cap N(o', A) = \emptyset$ for $o \neq o' \in O$.
    % \end{remark}

    \subsubsection{Partitioning of the optimal solution}
    For each $i \in \mathcal{J}$, let $I_i = (m_{i+1}, m_i]$ denote the $\nth{i}$ weight bucket. Throughout the rest of proof, we denote by $A_i$ the algorithm's solution right after $\textsc{Interval } \LS(A, (m_{i+1}, m_{i}])$. Similarly, we let $A'_i$ be the algorithm's solution right after $(m_{i+1}, \alpha)$-$\LS(A)$. 
    Let $L$ be the last index in $\mathcal{J}$. The final output of \Cref{alg:merge-LS-setpacking} is equivalently denoted by $A_{L+1}$ and $A'_L$.
    Finally, let $O_i \triangleq O \cap I_i$ be the optimal edges in the $\nth{i}$ weight bucket. Next, we define the edges of $\Orem \subseteq O$ whose neighborhood gets removed.
    We say an edge $e$ is \textit{evicted in round $j$} if we have that $e \in N$ for some $(m_j, \alpha)$-swap $(S,N)$. We denote by $R$ the set of all evicted elements. Then, we define
    %We can now define the subset of $O$ whose neighborhood gets evicted at some step in \Cref{alg:merge-LS-setpacking}.
    \begin{itemize}
        \item  Let $\Orem_i \triangleq \{ o \in O_i : \exists a \in N(o, A_i)\cup N(o, A'_i) \text{ that is evicted at some round }j \geq i\}.$ 
        \item Let $\Orem \triangleq \displaystyle\bigcup_{i \in \mathcal{J}} \Orem_i$.
    \end{itemize}
We are then left with the edges of $O$ whose neighborhood is contained in all iterations up to the final solution $A_{L+1}$. Notice that for $o \in O_i\setminus \Orem_i$, $N(o, A_i) \subseteq N(o, A'_i)$ as we assume the neighborhood is not evicted in round $i$ itself. For these edges, we will split the analysis into \emph{singles, doubles, and triples}. For each $i \in \mathcal{J}$, these sets are defined as:
    \begin{itemize}
        \item $S_i\triangleq \{ o \in O_i\setminus \Orem_i \colon \card{N_v(o, A_i)} = 1, c(N_v(o, A_i)) \leq 3 w_o\}$ be the set of singles in $O_i$ w.r.t $A_i$ whose neighbor is of similar weight.

        \item $D_i \triangleq \{ o \in O_i\setminus \Orem_i \colon \card{N_v(o, A_i)} = 2 , c(N_v(o, A_i)) \leq 3 w_o\}$ be the set of doubles in $O_i$ w.r.t $A_i$ whose neighbors are of similar weight.

        \item $T_i \triangleq \{ o \in O_i\setminus \Orem_i \colon \card{N_v(o, A_i)} \geq 3 \} \cup \{ o \in O_i \setminus \Orem_i \colon \card{N_v(o, A_i)} \leq 2, c(N_v(o, A_i)) > 3 w_o\} $ be the set of triples in $O_i$ w.r.t $A_i$ as well as items with very large weighted neighbors.
        \end{itemize}
    The same set names without the subscript $i$ correspond to the union over all intervals $i \in \mathcal{J}$, i.e. $T =\bigcup_{i \in \mathcal{J}} T_i$.
    In the next part of the proof, we show that $\Orem$ is only $o(k)$-larger than our final solution $A_{L+1}$ for a large choice of $\alpha$.

    \subsubsection{Negligible contribution of the edges with removed neighborhoods}
        Let $R$ be the set of all evicted edges in $\bigcup_{i}A_i \cup \bigcup_i A'_i$, i.e. $R = \{a \in \bigcup_{i}A_i \cup \bigcup_i A'_i: a \notin A_{L+1}\}$.. Observe that edges can only exit the solution when calling $(m_{i}, \alpha)$-$\LS$ for some $i$. Indeed, in any round $i \in \mathcal{J}$, we execute \textsc{Interval} \LS$(A, (m_{i+1}, m_i])$ and the interval $(m_{i+1}, m_i]$ is a new, unprocessed interval such that \textsc{Interval} \LS$(A, (m_{i+1}, m_i])$ only adds elements into the current solution.
        Thus $A'_{i-1} \subseteq A_i$ where we denote by $A'_{i-1}$ as the solution just prior to processing index $i \in \Jcal$. Similarly to the analysis for matroid $k$-parity, we bound the contribution of edges that are evicted throughout the algorithm in favor of improving swaps. 
        We use the following lemma over the weight of the evicted edges. 

    \begin{mylemma}
        \label{lem:discard-elements}
        Let $R$ be the set of evicted edges, and $A_{L+1}$ be the final solution. Then, we have
        \begin{align*}
            w(R) & \leq \frac{1}{\alpha - 1} \cdot w(A_{L+1}).
        \end{align*}
    \end{mylemma}
    \begin{proof}
        [Proof of \Cref{lem:discard-elements}]
        The proof can be directly obtained by \Cref{lem:bound-X}. We restate the proof removing notations related to submodularity and appearances (\Cref{def:appearance}).
        We look at the increase in the solution value after an $(m, \alpha)$-improvement. Let $A$ be the current solution. Consider an $(m, \alpha)$ improvement $(S, N)$ for the current solution $A$. Let $A' \triangleq (A \setminus N) \cup S$ be the solution obtained by applying the improvement. Then, we have that
        \begin{align*}
            w(A') - w(A) & = w((A \setminus N) \cup S) - w(A)
            =  w(S) - w(N)
            \geq (\alpha - 1) w(N).
        \end{align*}
        We conclude the proof by summing over all exchanges performed throughout the execution of the algorithm. Let $T$ be the number of $(m_i, \alpha)$-swaps for $i \in \mathcal{J}$. Let $A^{(t)}$ be the current solution during the $\nth{t}$ exchange, where $N_t$ is discarded.
        Now, observe that each $e \in R$ appears in at least of the set $N_t$ for some $ t \in [T]$.
        Then, we have that
        \begin{align*}
            (\alpha - 1) \sum_{e \in R} w_e & \leq (\alpha - 1) \sum_{t \in [T]} w(N_t) \leq \sum_{t \in [T]} \ld w(A^{(t)}) - w(A^{(t-1)}) \rd =  w(A_{L+1}) \enspace. \qedhere
        \end{align*}
    \end{proof}

    The next lemma shows that, assuming $\alpha = \Omega({1})$, the weight of $\Orem$ is negligible (an $o(k)$-factor) compared to the weight of the final solution.
    \begin{mylemma}
        \label{lem:removed-setpack}
        The following equation holds:
        \begin{align*}
            w(\Orem) & \leq \frac{2k}{(\alpha - 1)} \cdot w(A_{L+1})\enspace.
        \end{align*}
    \end{mylemma}
    \begin{proof}
        [Proof of \Cref{lem:removed-setpack}]
        Let $o \in \Orem_i$. Then, there exists some $a \in N(o, A_i) \cup N(o, A_i')$ evicted at some round $j \geq i$. In particular, we have that $a \in R$, and $w_a \geq 1/2 w_o$ since $w_o \in (m_{i+1}, m_i]$ and $w_a \geq m_{i+1}$.
        For each $i \in \mathcal{J}$ and $o \in \Orem_i$, let $a \triangleq \pi(o)$ be one arbitrary such evicted neighbor.
        For each $a \in R$, at most $k$ edges $o \in \Orem$ have $a = \pi(o)$ as its neighbor since $a$ has $k$ vertices and $o \in \Orem$ can only conflict with distinct vertices each. Therefore, 
        \begin{align*}
            kw(R) & = k \displaystyle\sum_{a \in R} w_a \geq \displaystyle\sum_{o \in \Orem} w(\pi(o)) \geq \frac{ w(\Orem)}{2}.
        \end{align*}
        The proof terminates by applying \Cref{lem:discard-elements}.
        \end{proof}

        It is worth noting that \Cref{lem:removed-setpack} distinctly works for a $k$-set packing instance as opposed to a general matroid $k$-parity. This is because any $o \in \Orem$ can only conflict with distinct vertices over all time steps, whereas matroid exchanges per time step do not guarantee this property.

\subsubsection{Negligible contribution of singles}
    The next lemma shows that $A_{L+1}$ is a good approximation for the subset of the singles $S$ that are still singles after $\alpha$-$\LS$ is performed.
    To show this, we further partition the sets $S_i$ for each $i \in \mathcal{J}$ as follows:
    \begin{enumerate}
        \item $\Snear_i \triangleq \{ o \in S_i: N(o, A_i) \subseteq (m_{i+1}, m_i]\}$. 
        \item $\Sfarst_i \triangleq \{ o \in S_i: N(o, A_i) \not\subseteq (m_{i+1}, m_i], N(o, A_i) = N(o, A'_i)\}$.
        \item $\Sfarunst_i \triangleq \{ o \in S_i: N(o, A_i) \not\subseteq (m_{i+1}, m_i], N(o, A_i) \neq N(o, A'_i)\}$.
    \end{enumerate}

    We first show that $A_{L+1}$ is a $o(k)$ approximation relative to $\Snear \cup \Sfarst$. We will then defer the analysis over $\Sfarunst$ to Section \ref{sec: trip, d far, s unstable}.
    
    \begin{mylemma}
        \label{lem: bound for A_{L+1}}
        The final solution $A_{L+1}$ satisfies 
        $$w(\Snear\cup\Sfarst) \leq 8 \alpha \cdot w(A_{L+1}).$$
    \end{mylemma}

    \begin{proof}[Proof of \Cref{lem: bound for A_{L+1}}]
    To prove this relation, we claim that for any $a \in A_{L+1}$, 
    \begin{align*}\alpha w_a &\geq w(\{ o \in \Snear_i: N(o, A_i) = a\}) \quad \text{ and } \\ \quad\alpha w_a &\geq w(\{ o \in \Sfarst_i: N(o, A_i) = a\}).\end{align*}
    We show the first inequality. Fix any $i \in \mathcal{J}$ and consider $o \in \Snear_i$. 
    Notice that since $N(o, A_i)$ is not evicted after round $i$, we know that $N(o, A_i) = \{a\} \in A_{L+1}$. 
    For any $a \in A_{L+1}$, there can be at most one edge $o \in \Snear_i$ such that $N(o, A_i) = \{a\}$. Indeed, suppose by contradiction that there exists $o_1, o_2 \in \Snear_i$ with neighborhood $\{a\}$. Then, because $N(o_p, A_i) \subseteq (m_{i+1}, m_i]$ for $p \in \{1, 2\}$, we have that $N(o_p, A_{i-1}') = \emptyset$, so both edges are considered by $\textsc{Unweighted } \LS$ when applying $\textsc{Interval } \LS(A_{i-1}, (m_{i+1}, m_i])$. In particular, the swap $A_i - a + \{ o_1, o_2\}$ would be improving in size, which is a contradiction to termination of $\textsc{Interval } \LS(A_{i-1}, (m_{i+1}, m_i])$. This implies that 
    \begin{align*}
        \alpha w_a \geq \alpha m_{i+1} &= \alpha\frac{m_i}{2} \cdot 1 
        \geq \alpha\frac{m_i}{2} \cdot |\{o \in \Snear_i: N(o, A_i) = a\}|
        \geq w(\{o \in \Snear_i: N(o, A_i) = a\}),
    \end{align*} 
    where the last inequality follows from $w_o \leq m_i$ for $o \in S_i$ and assuming $\alpha \geq 2$.\\

    We then show the second inequality. First, notice that by definition, every $o \in \Sfarst_i$ satisfies $\card{N_v(o, A'_i)} =1$ and $N(o, A'_i) \subseteq A_{L+1}$ as the edge does not get evicted.
    Then fix any $a \in A_{L+1}$ and notice that the swap $A'_i - a + \{ o \in \Sfarst_i: N(o, A'_i) = a\}$ is a feasible swap for $A'_i$. By termination of $(m_{i+1}, \alpha)$-$\LS$ with solution $A'_i$, the swap cannot be $\alpha$-improving such that $\alpha w_a \geq w(\{ o \in \Sfarst_i: N(o, A'_i) = a\}) = w(\{ o \in \Sfarst_i: N(o, A_i) = a\})$. We are now able to prove the lemma using these two inequalities. \\

    Notice that for any $a \in A_{L+1}$, and any $o \in S_i$ such that $N(o, A_i) = a$, the weight of $a$ must satisfy $w_a \leq 3 w_o$ by definition of $S_i$. Lastly, since $o \in S_i \subseteq (m_{i+1}, m_i]$, we have $w_o \leq m_i = W\tau2^{-i}$. Then,
    \begin{align*}
        w_o & \leq W\tau2^{-i} 
        = \frac{(W\tau2^{-i}) w_a}{w_a}
        \leq \frac{(W\tau2^{-i}) 3w_o}{w_a}.
    \end{align*}
    A reordering shows that $i \leq  \log_{2}\left( \frac{3W\tau}{w_a} \right)$. Moreover, we know that $w_o \leq m_i = 2m_{i+1}\leq  2w_a$ and $w_o \geq m_{i+1} = W\tau2^{-(i+1)}$ so that $i \geq \log_{2}(\frac{W\tau}{4w_a})$. Thus for every $o \in S$ that has $a$ as a neighbor,  the edge $o$ must lie within a bucket $i \in [l_a, r_a]$ where $l_a = \log_{2}\left( \frac{W\tau}{4w_a} \right)$ and $r_a = 
    \log_{2}(\frac{3W\tau}{w_a})$.
    Then,
    \begin{align*}
        w(\Snear \cup \Sfarst) &= \sum_{a \in A_{L+1}}\sum_{i\in \mathcal{J}} w(\{ o \in \Snear_i \cup \Sfarst_i: N(o, A_i) = a\}) \\
        & \leq \sum_{a \in A_{L+1}}\sum_{i \in [l_a, r_a]} w(\{ o \in \Snear_i : N(o, A_i) = a\}) + w(\{ o \in \Sfarst_i: N(o, A_i) = a\})\\
        & \leq \sum_{a \in A_{L+1}}\sum_{i \in [l_a, r_a]} 2\alpha w_a \\
        & = \sum_{a \in A_{L+1}} \left(\log_{2}\left( \frac{3W\tau}{w_a}\right) - \log_{2}\left(\frac{W\tau}{4w_a}\right) \right) 2\alpha w_a \\
        & = \log_{2}(12) \cdot 2\alpha w(A_{L+1}) \\
        & \leq 8\alpha w(A_{L+1}),
    \end{align*}
    where the first equality follows by definition of $S$ and $N(o, A_i) \subseteq A_{L+1}$ for $o \in S_i$ concluding the lemma. \qedhere
 \end{proof}

\subsubsection{Analyzing the contribution of doubles}
    We show that $A_{L+1}$ has a good approximation guarantee relative to the set of doubles $D$. 
    As in the previous part, we separate the analysis into opt-edges whose neighborhood is within their own weight class and those that contain a neighbor in a prior class. 
    The key is that the former case can be analyzed via an unweighted analysis akin to \cite{Cygan:2013:Improved}, while the latter can be analyzed according to the random weight bucketing argument from \cite{Singer-Thiery:2025:Better}. 
    We therefore split $D_i$ into two different sets as follows:
    \begin{enumerate}
        \item Let $\Dnear_i \triangleq \{o \in D_i : N(o, A_i) \subseteq (m_{i+1}, m_i]\}$.
        \item Let $\Dfar_i \triangleq \{o \in D_i: N(o, A_i) \not\subseteq (m_{i+1}, m_i]\}$.
    \end{enumerate}

We first show that $A_{L+1}$ is an $O(1)$ approximation relative to $\Dnear$. 
\begin{mylemma}\label{lem: k/3 approx for doubles}
    The final solution $A_{L+1}$ satisfies 
    $$\frac{1}{2(1+\e)} w(\Dnear)\leq w(A_{L+1}) \enspace.$$
\end{mylemma}

\begin{proof}[Proof of \Cref{lem: k/3 approx for doubles}]
    At any iteration $i \in \mathcal{J}$, $\textsc{Interval Local-Search}$ applies Cygan's algorithm \cite{Cygan:2013:Improved} ($\textsc{Unweighted } \LS$ in \Cref{alg:interval-local-search}) on a set of edges $E'$ (treated as unweighted) and outputs a solution $A'$.
    We use the following claim on the output $A'$ paraphrased from \cite{Cygan:2013:Improved}:

\begin{mylemma}[\cite{Cygan:2013:Improved}]\label{claim: logn swap size}
    Let $W$ be any feasible solution, and let $A'$ be the output of the bounded pathwidth local search algorithm of \cite{Cygan:2013:Improved} with parameter $\e$. Let $W_s \triangleq \{w \in W: |N_v(w, A')| = 1\}$ and $W_d \triangleq \{w \in W: |N_v(w, A')| = 2\}$. Then, the following relations hold:
    \begin{enumerate}
        \item $|W_s| \leq \e |N(W, A')|$,
        \item $|W_d| \leq (1+\e) |N(W, A')|$.
    \end{enumerate}
\end{mylemma}

\Cref{claim: logn swap size} shows that a single run of the \textsc{Interval Local Search} (Cygan's algorithm) will result in a small amount of singles and doubles over that unweighted sub-instance. We will use this to show that $\Dnear$ is not much larger than its neighborhood.
For each $i \in \mathcal{J}$, define the following set of neighbors of $\Dnear_i$:
 $$\Anear_i \triangleq \{ a \in A_{L+1}: a \in N(\Dnear_i, A_i)\}.$$

We claim that $|\Anear_i| \geq \frac{1}{1+\e}|\Dnear_i|\label{eq: size A lasting}.$ To see that this is true, let $\Dnear_i$ be a feasible set packing solution and note that $A_i$ is the output of $\textsc{Interval Local Search}(A'_{i-1}, (m_{i+1}, m_i])$ which runs the algorithm of \cite{Cygan:2013:Improved} over the interval $I_i = (m_{i+1}, m_i]$ where notice that $A'_{i-1}\cap I_i = \emptyset$. Since $\Dnear_i$ is a subset of edges that are contained in interval $I_i$ and whose neighborhood is contained in $I_i$, then $\Dnear_i$ is a feasible solution of the algorithm $\textsc{Interval Local Search}(A'_{i-1}, (m_{i+1}, m_i])$. Moreover, since $|N_v(o, A_i)| = |N_v(o, \Anear_i)| =2$ for each $o \in \Dnear_i$, then applying \Cref{claim: logn swap size} with $W$ as $\Dnear_i$ and $N(\Dnear_i, A_i) = \Anear_i$ shows that $|\Dnear_i| \leq (1 + \e)|N(\Dnear_i, A_i)| = (1+\e) |\Anear_i|$ as desired.
Then since $(1+\e)|\Anear_i| \geq |\Dnear_i|$, we get 
\begin{align*}
    w(\Anear_i) & \geq m_{i+1} |\Anear_i| 
    \geq \frac{m_{i+1}}{1+\e} |\Dnear_i| 
    \geq  \frac{m_{i+1}}{m_{i}(1+\e)} w(\Dnear_i),
\end{align*}
so that for each $i \in [L]$, we have $w(\Anear_i) \geq \frac{1}{2(1+\e)} w(\Dnear_i)$. We then conclude by noting that for each $i \in [L]$, $\Anear_i$ is a disjoint set so that $w(\Anear) = \sum_{i \in [L]} w(\Anear_i) \geq \sum_{i \in [L]} \frac{1}{2(1+\e)}w(\Dnear_i) = \frac{1}{2(1+\e)} w(\Dnear)$. The proof finishes by observing that by definition of $O_i$, the edges satisfy $\Anear \subseteq A_{L+1}$.  \qedhere
    \end{proof}

\subsubsection{Approximation guarantee of $T$, $\Dfar$, and $\Sfarunst$}
\label{sec: trip, d far, s unstable}

Finally, it remains to analyze the triples $T$, the doubles in $\Dfar$ that contain a neighbor in a prior interval, and the singles in $\Sfarunst$ whose neighborhood changes during $\alpha$-$\LS$. For each $o \in T \cup \Dfar \cup \Sfarunst$, we denote by $m_o$ the closest marker to $w_o$ such that $m_o \geq w_o$.
The contribution of these edges depends on the distance between $w_o$ and $m_o$ which can be computed directly from the distribution $\tau$.
Towards this end, for every $o \in O$ we define the ratio of the item's weight to its nearest marker as $r_o \triangleq \frac{m_o}{w_o}$.
Intuitively, the larger the ratio is, the better the approximation. As shown in the next lemma, the guarantees directly depend on $r_o$.
    \begin{mylemma}
        \label{lem:improved-gap-triples}
        For any $i \in \mathcal{J}$ and any $o \in (\Dfar_i \cup T_i \cup \Sfarunst_i)$, we have
        \begin{align*}
            1.5r_o\cdot w_o & \leq c(N_v(o, A_i)) \quad \text{ for }o \in \Dfar_i \cup T_i \\   1.5r_o\cdot w_o  & \leq c(N_v(o, A'_i))\quad \text{ for }o \in \Sfarunst_i.
        \end{align*}
    \end{mylemma}

    \begin{proof}[Proof of \Cref{lem:improved-gap-triples}]
        Suppose first that $o \in T_i$ for some $i \in \mathcal{J}$. By definition of $T_i$, we have that either $c(N_v(o, A_i)) \geq 3 w_o \geq 3m_{i+1}$ or $c(N_v(o, A_i)) \geq 3 \cdot \min_{a \in A_i}w_a \geq 3m_{i+1}$.
        Since $o \in I_i$, the nearest marker is $m_{i}$ such that $m_o = m_{i}$. This means that $w_o = \frac{m_{i}}{r_o}$ by definition of $r_o$. Since $w_o \geq m_{i+1}$, we get that
            \begin{align*}
                w_o & = \frac{w_o}{c(N_v(o, A_i))}  c(N_v(o, A_i))
                \leq  \frac{m_{i}}{3 m_{i+1}r_o} c(N_v(o, A_i))
                = \frac{2}{3 r_o} c(N_v(o, A_i)),
            \end{align*}
            where the last equality follows from $m_{i+1}/m_{i} = 1/2$ for any $i$. This proves that \Cref{lem:improved-gap-triples} holds for $o \in T$ by rearranging.\\

        Next suppose that $o \in \Dfar_i$ for any $i \in \mathcal{J}$. By definition, any $o \in \Dfar_i$ satisfies $N(o, A_i) \not\subseteq (m_{i+1}, m_{i}]$ and $\card{N_v(o, A_i)} = 2$. Without loss of generality, let $v_1, v_2$ be the two vertices in $N_v(o, A_i)$ such that $c(v_1) \geq m_{i}$ and $c(v_2) \geq m_{i+1}$. Since $o \in I_i$, the nearest marker is $m_{i} = m_o$. This means that $w_o = \frac{m_{i}}{r_o}$ and $w_o \geq m_{i+1}$.
            Therefore,
            \begin{align*}
                w_o 
                & =  \frac{m_{i}}{r_o(c(v_1) + c(v_2))} c(N_v(o, A_i)) 
                \leq  \frac{m_{i}}{r_o(m_i + m_{i+1})} c(N_v(o, A_i))
                = \frac{1}{1.5r_o} c(N_v(o, A_i)),
            \end{align*}
            where the last equality follows from $m_{i+1}/m_{i} = 1/2$ for $i \in \mathcal{J}$. This proves that \Cref{lem:improved-gap-triples} holds for $o \in \Dfar$.\\

            Finally, suppose that $o \in \Sfarunst_i$ for any $i \in \mathcal{J}$. By definition, any $o \in \Sfarunst_i$ satisfies $N(o, A_i) \neq N(o, A'_i)$. Moreover, since $N(o, A_i)$ is never evicted in round $i$ and onwards, it must be that $N(o, A_i) \subset N(o, A'_i)$ such that $\card{N_v(o, A'_i)}\geq 2$. In particular, let $v_1, v_2$ be two vertices in $N_v(o, A'_i)$ with $v_1 \in N_v(o, A_i)$. By definition of $\Sfarunst_i$, we have $c(v_1) \geq m_i$ since $N(o, A_i)\not\subseteq (m_{i+1}, m_i]$. Therefore, identically to the proof of $o \in \Dfar_i$, we get that 
            \begin{align*}
                w_o 
                & \leq  \frac{m_{i}}{r_o(c(v_1) + c(v_2))} c(N_v(o, A'_i)) 
                \leq  \frac{m_{i}}{r_o(m_i + m_{i+1})} c(N_v(o, A'_i))
                = \frac{1}{1.5r_o} c(N_v(o, A'_i)),
            \end{align*}
            concluding the lemma.
    \end{proof}
    Note that for every $o \in O$, $r_o$ is a random value as it depends on the outcome $\tau$ which determines the markers $m_0, \hdots, m_{L+1}$. We would like most elements to have a large ratio $r_o$ as this guarantees that the factor difference in weight to its neighborhood accounts for a large multiplicative improvement. 
    This fact directly follows from Feldman and Ward \cite{Feldman-Ward:2026:Submodular}.
    \begin{mylemma}[\protect{\cite[Lemma~B.1]{Feldman-Ward:2026:Submodular}}]\label{lem: prob of Rj}
        For any $o \in O$, $\mathbb{E}[r_o] = \ln^{-1}(2)$.
    \end{mylemma}

    Combining \Cref{lem: prob of Rj} and \Cref{lem:improved-gap-triples}, we bound the value of $\Dnear \cup T \cup \Sfarunst$ relative to $A_{L+1}$.
    \begin{mylemma}\label{lem: main apx triples and goods}
        Let $G \triangleq \Dfar \cup T \cup \Sfarunst$. Then, the following equation holds:
        \begin{align*}
            \esp{w(G)} \leq \frac{\ln(4)}{3}\cdot k \cdot \mathbb{E}[w(A_{L+1})] + \esp{w(O \setminus G)}.
        \end{align*}
    \end{mylemma}
    \begin{proof}[Proof of \Cref{lem: main apx triples and goods}]
        For simplicity, denote $G \triangleq \Dfar \cup T \cup \Sfarunst$. For any $o \in O$, let $H_o \triangleq \bm{1}\lc o \in G \rc$ be the indicator of whether $o \in G$.
        Observe that $\esp{w(G)} = \sum_{o \in G} w_o \esp{H_o}$. We would like to apply \Cref{lem: prob of Rj} but need to be careful due to correlation between $H_o$ and $r_o$. We use the following inequality:
        \begin{align}
            \esp{H_o} & \leq \ln(2) \esp{r_o H_o} + \esp{1 - H_o}. \label{eq:correlation}
        \end{align}
        This is a valid inequality, since it is equivalent to $\esp{H_o(2 - \ln(2)r_o)} \leq 1$, which is true because $H_o \leq 1$ and $2 - \ln(2) r_o \geq 0$ for all $r_o \in [1, 2)$. In particular, we have
        \begin{align*}
            \esp{H_o(2 - \ln(2)r_o)} & \leq \esp{2 - \ln(2) r_o} = 2 - 1 = 1,
        \end{align*}
        where we used \Cref{lem: prob of Rj} for the second equality. 
        Applying \Cref{eq:correlation}, we get that
        \begin{align*}
            \esp{w(G)} & = \sum_{o \in G} w_o \esp{H_o} \leq \ln(2) \sum_{o \in G} w_o \esp{r_o H_o} +  \sum_{o \in G} w_o \esp{1 - H_o} \\
            & = \ln (2) \esp{ \sum_{o \in G} w_o r_o} + \esp{\sum_{o \notin G} w_o} \\
            & = \ln (2) \esp{ \sum_{o \in G} w_o r_o} + \esp{w(O \setminus G)}.
        \end{align*}
        Next use \Cref{lem:improved-gap-triples} to bound the first term on the right-hand side.
        By \Cref{lem:improved-gap-triples}, for every $i \in \mathcal{J}$ and for any $o \in T_i\cup \Dfar_i$, we have that $1.5 r_o\cdot w_o \leq c(N_v(o, A_i))$ and for $o \in \Sfarunst_i$, we have that $1.5 r_o\cdot w_o \leq c(N_v(o, A'_i))$. Summing this relation over $i \in \mathcal{J}$, we get that 
        \begin{align*}
            \displaystyle\sum_{o \in G} 1.5r_o w_o & \leq \displaystyle\sum_{i \in \mathcal{J}} \displaystyle\sum_{o \in T_i\cup\Dfar_i} c(N_v(o, A_i)) + \displaystyle\sum_{o \in \Sfarunst_i} c(N_v(o, A'_i)) \\
            & \leq \displaystyle\sum_{i \in \mathcal{J}}\displaystyle\sum_{o \in G_i} c(N_v(o, A_{L+1}) ) \\
            & \leq c(v(A_{L+1}))\\
            &\leq kw(A_{L+1})
        \end{align*}
        where the second inequality follows from the fact that $N_v(o, A_i) \subseteq N_v(o, A_{L+1})$ and $N_v(o, A'_i) \subseteq N_v(o, A_{L+1})$ for every $o \in O\setminus \Orem$. The third inequality uses that $N_v(o, A_{L+1}) \cap N_v(o', A_{L+1}) = \emptyset$ for all $o' \neq o$. The last inequality is due to the fact that every hyperedge contains at most $k$ vertices. 
        Combining the equations together, we finally obtain that
        \begin{align*}
            \esp{w(G)} & \leq \frac{\ln(4)}{3}\cdot k \cdot \mathbb{E}[w(A_{L+1})]  + \esp{w(O\setminus G)}. \qedhere
        \end{align*}
\end{proof}

\subsubsection{Putting the pieces together}
We are now able to conclude our main theorem which shows that for a choice of $\alpha$ an arbitrarily large constant, \Cref{alg:merge-LS-setpacking} improves over Neuwohner's $\frac{k}{2.0056} + O(1)$-approximation algorithm.
\setpackapprox*
    \begin{proof}[Proof of \Cref{thm:APX-better}]
        Let $G \triangleq \Dfar \cup T \cup \Sfarunst$. For any outcome of $\tau$, we have that
        \begin{align}
            w(O) & = w(\Orem) + w(\Snear \cup \Sfarst) + w(\Dnear) + w(\Dfar \cup T \cup \Sfarunst)  \notag
        \end{align}
    since these sets are defined to partition $O$.
    We first bound the first 3 terms. By \Cref{lem:removed-setpack}, \Cref{lem: bound for A_{L+1}} and \Cref{lem: k/3 approx for doubles}, 
    \begin{align*}
        w(O \setminus G) = w(\Orem) + w(\Snear \cup \Sfarst) + w(\Dnear) & \leq \left( \frac{2k}{(\alpha-1)} + {8\alpha} + 2(1+\e)\right) w(A_{L+1}) %\label{eq: approx all other}
    \end{align*}
    and for $\alpha = \Theta(\sqrt{k})$, this shows that 
    \begin{align*}
        w(\Orem) + w(S) + w(\Dnear) & \leq O(\sqrt{k}) w(A_{L+1}). \label{eq: approx all other}
    \end{align*}
    Plugging in \Cref{lem: main apx triples and goods} with an upper bound for $w(G)$ yields the result.
    \end{proof}

    \section{Conclusion and Open Questions}
    The main contribution of this paper is a new, arguably simple algorithm for monotone submodular maximization over a matroid $k$-parity constraint. It applies a local-search procedure to edges of weights above some threshold, which is decreased continuously, so that smaller value edges are taken in favor of large value edges only if they improve the solution greatly.
    The guarantee of our algorithm asymptotically matches that of the state-of-the-art algorithm in the unweighted setting. We obtain a $\frac{k}{2} + o(k)$-approximation for monotone submodular maximization subject to a matroid $k$-parity constraint. This improves over the $\frac{\ln(4)k}{1 + \ln(2)} + o(k)$-approximation by Feldman and Ward in the same setting, and approaches, albeit an additive $o(k)$-loss, the $\frac{k}{2} + \e$-approximation guarantee by Lee, Sviridenko, and Vondr{\'{a}}k in the unweighted setting.
    We use this algorithm in the special case of weighted $k$-Set Packing. We show that a combination of our algorithm and that of Singer and Thiery  \cite{Singer-Thiery:2025:Better} allows us to obtain a $\frac{\ln(4)k}{3}+o(k)$-approximation. This improves over the prior state of the art  $\frac{k}{2.00561}+O(1)$-approximation algorithm \cite{Neuwohner:2023:Passing}. \\

    Our work leaves some interesting open questions to explore. Our results demonstrate that it might be easier to obtain improved guarantees for large values of $k$. Can one also obtain a $k/3 + o(k)$-approximation for weighted $k$-Set Packing to match the unweighted guarantee? 
    Another natural promising avenue is to understand low values of $k$. For (weighted) $3$-matroid intersection, the state-of-the-art approximation guarantee is $2$ (due to a $k-1$ approximation for all $k$) by Lee et al. \cite{Lee:2010:Submodular}. For $3$-Set Packing, this bound can be surpassed \cite{Neuwohner:2021:Improved, Thiery:2023:Approximation, Thiery:2023:Approximation} but is still far from the $3/2$ approximation in the unweighted setting, which is a natural barrier. Improving on either of these results would be very interesting.
    Finally, we ask whether there are \emph{black-box} reductions from the weighted to unweighted setting as done recently in other computational models over matching and matroid intersection (see e.g. \cite{Bernstein-Dudeja-Langley:2021:Framework,Huang-Kakimura-Kamiyama:2019:Exact,Dudeja-Grilnberger:2026:Unweighted}).

    \section*{Acknowledgements}
        We thank Neil Olver who participated in many of the discussions and progress that led up to and inspired the present algorithm.
    
    \section*{Tool and Computational Resource Disclosure}
        All the core ideas in this paper, including the algorithms and the proof strategies, were developed without the use of Large Language Models. Chat-GPT was used to obtain a simplification of a complex case-by-case analysis that the authors originally employed. This simplification was adapted into \Cref{lem:single-occurences}, which was independently extended and verified by the authors. All writing was done by the authors.

    \bibliography{biblio}
    \bibliographystyle{alpha}
    \appendix
    \section{Matroid Exchanges}
    \label{sec:appendix-matroid}
    We analyze the quality of the algorithm's solution by constructing non-improving local swaps between the intermediary solution $A_\tau$ for any $\tau \in [\btheta, W]$ and the optimal solution $O$.
    It will be helpful to define matroid exchanges over the vertices covered by solutions $A_\tau$ and $O$, before constructing feasible swaps between hyperedges of $A_\tau$ and $O$.
    \begin{restatable}{mylemma}{nestedsets}
        \label{thm:nested-sets}
            Let $A_\tau, X_\tau$ for any $\tau \in [\btheta, W]$. There exists a set $B \subseteq v(O)$ such that $|B| = |v(A_\tau)|$ with the property that $v(A_\tau) \cup (v(O) \setminus B) \in \cI$.
            Moreover, for any $o \in O$, such that $f(o \mid X_\tau - o) \geq \tau$, we have that $v(o) \cap B \neq \emptyset$.
    \end{restatable}
    The set $B$ should be thought of as the set of vertices in $v(O)$ that are conflicting with $A_\tau$.
    \begin{proof}[Proof of \Cref{thm:nested-sets}]
            Without loss of generality, we may assume that, we have $\card{v(O)} \geq \card{v(A_\tau)}$ by adding disjoint edges with $0$ marginal contribution, each containing $k$ new dummy vertices independent with the rest of $O$. Then, the first part of the proposition follows from the matroid extension property.
            To prove the second statement, we suppose towards contradiction that $v(o) \cap B = \emptyset$. This implies that:
            \begin{align*}
                v(A_\tau) \cup v(o) & \subseteq v(A_\tau) \cup v(o) \cup ( v(O) \bb B ) = v(A_\tau) \cup ( v(O) \bb B ) \in \cI.
            \end{align*}
            So $v(o)$ can be added to $v(A_\tau)$ without violating the independence constraint. Since $f(o \mid X_\tau-o) \geq \tau$, the pair $(\{o\}, \emptyset)$ is a $(\tau, \alpha)$-improvement for $A_\tau$. This contradicts \Cref{claim:no tau swaps}.
    \end{proof}
    We also restate known results about matroid-basis exchanges
    \begin{restatable}{mythm}{Rota-setminus}[Proposition 5 in \cite{Lason:2015:List-Coloring}]
            \label{thm:Rota-setminus}
            Let $A$ and $B$ be two bases of a matroid $\cM$. Then for any partition $B_1, \sqcup \ldots \sqcup B_m = B$ there exists a partition $A_1 \sqcup \ldots \sqcup A_m = A$ such that
            \begin{align*}
                (B \setminus B_i) \cup A_i \mbox{ are bases for all } i \in [m].
            \end{align*}
        \end{restatable}

    \begin{restatable}{mythm}{GeneralizedRota}
    [Proposition 6 in \cite{Lason:2015:List-Coloring}]
            \label{thm:Rota-non-basis-2}
            Let $A$ and $B$ be two bases of a matroid $\cM$.
            For any partition $B_1, \ldots, B_m$ of $B$, there are disjoint sets $A_1, \ldots, A_m \subseteq A$ such that 
            \begin{align*}
                (A \setminus A_i) \cup B_i & \mbox{ are bases for all } i \in [m].
            \end{align*}
        \end{restatable}

    \subsection{Constructing a Conflict Graph}
    We use the above matroid exchanges to construct a bipartite conflict graph between $O_{\geq \tau}$ to $A_\tau$. The conflict graph is comprised of feasible hyperedge swaps between $O_{\geq \tau}$ to $A_\tau$ and therefore satisfies local swap guarantees. We restate the lemma here.
    \ConflictGraph*
    \begin{proof}[Proof of \Cref{thm:decomposition-ith-round-2}]
        The proof proceeds by finding a decomposition of $A_\tau$ and $O$ into swaps. We will obtain a collection of pairs $\{(o, N_o)\}_{o \in O}$ where $v(o)$ and $N_o$ form a matroid exchange and $\{N_o\}_{o \in O}$ partitions $v(A_\tau)$. We then simply construct the bipartite conflict  graph by placing an edge between $o$ and each $v \in N_o$. 
        By \Cref{thm:nested-sets}, there exists a set $B \subseteq v(O)$ of vertices conflicting with $A_{\tau}$ such that $v(A_\tau) \cup (v(O) \setminus B) \in \cI$. 
        By \Cref{thm:nested-sets}, each $o \in O_{\geq \tau}$ appears in either $O_s$ or $O_d$.
        Let $O_s \triangleq \{ o \in O_{\geq \tau} \colon |v(o) \cap B| = 1\}$ and $O_d \triangleq \{ o \in O_{\geq \tau} \colon |v(o) \cap B| > 1\}$ and $O_b \triangleq O \setminus (O_s \cup O_d)$ be the sets of single conflicts, double conflicts, and remaining edges. 
        Let $\cM'$ be the matroid contracted on $v(O) \setminus B$. Clearly, the sets $v(A_\tau)$ and $B$ are both independent in $\cM'$.
        We partition $B = B_s \sqcup B_d \sqcup B_b$ into three sets where $B_s = v(O_s) \cap B$, $B_d = v(O_d) \cap B$ and $B_b = v(O_b) \cap B$. 
        By \Cref{thm:Rota-non-basis-2}, for each $q \in \{s, d, b\}$ there exists $N_q \subseteq v(A_{\tau})$ with the property that $|N_q| = |B_q|$ and $(v(A_\tau) \setminus N_q) \cup B_q$ is independent in $\cM'$.
        In particular, this means that $(v(A_\tau) \setminus N_q) \cup B_q \cup (v(O) \setminus B) \in \cI$.

        \paragraph{Constructing swaps over $\boldsymbol{O_s}$.}
        We construct a partition of $N_s$ by defining for every $a \in A_\tau$ the set $N_{s, a} \triangleq N_s \cap v(a)$ (some of which are empty). Let $\cM_s$ be the matroid contracted on $v(A_\tau) \setminus N_s \cup (v(O) \setminus B)$ and observe that $B_s$ and $N_s$ are both independent in $\cM_s$ with $|N_s| = |B_s|$. We apply \Cref{thm:Rota-setminus} with the partition $\{N_{s, a}\}_{a \in A_\tau}$ to obtain a partition of $B_s$ into sets $\{B_{s, a}\}_{a \in A_\tau}$ such that
        \begin{align*}
            (N_s \setminus N_{s, a}) \cup B_{s, a} & \in \cM_s
        \end{align*}
        so that $(v(A_\tau) \setminus N_{s, a}) \cup B_{s, a} \cup (v(O) \setminus B) \in \cI$ for any $a \in A_\tau$. This matroid exchange guarantees a feasible swap between the hyperedges $a$ and $e(B_{s, a})$, i.e. that $A_\tau - \{a\} + e(B_{s, a})$ is a feasible solution for every $a \in A_\tau$. Indeed, the underlying vertices form an independent set in $\Ical$ since
        \begin{align*}
            v(A_\tau) \setminus v(a) \cup v(e(B_{s, a})) & \subseteq v(A_\tau) \setminus N_{s, a} \cup v(e(B_{s, a})) \cup ( v(O) \setminus B) \\
            & \subseteq \lb v(A_\tau) \setminus N_{s, a} \cup B_{s, a} \rb \cup (v(O) \setminus B),
        \end{align*}
        where the last inequality follows from the fact $v(e(B_{s, a})) = (v(e(B_{s, a})) \cap B) \cup (v(e(B_{s, a})) \setminus B)$ and as $v(e(B_{s, a})) \cap B = B_{s, a}$, and $v(e(B_{s,a})) \setminus B \subseteq v(O) \setminus B$.
        Let $e(B_{s, a}) = \{o_1 \prec \ldots \prec o_\ell\}$ be the ordering of the edges in $e(B_{s, a})$ in $O$ (definition in \Cref{sec: analysis of algo 2}).
        We have that each $o_j \in e(B_{s, a})$ satisfies $\ow_\tau(o_j) \geq \tau$ by definition of $O_{\geq \tau}$. Thus, by submodularity, we have that $f(o_j \mid X_\tau - \{o_j\} \cup \{o_1, \ldots, o_{j-1}\}) \geq \ow_\tau(o) \geq \tau$ for all $j \in [\ell]$.
        By \Cref{claim:no tau swaps}, there are no improving $(\tau, \alpha, X_\tau)$-swaps on $A_\tau$, so that
        \begin{align}
            \sum_{o \in e(B_{s, a})} \ow_\tau(o) & \leq \sum_{j \in [|e(B_{s, a})|]} f(o_j \mid X_\tau - \{o_j\} \cup \{o_1, \ldots, o_{j-1}\}) \notag\\
            & \leq f(e(B_{s, a}) \mid X_\tau - e(B_{s, a})) \notag\\
            & \leq \alpha w_{\prec_\tau}(a, A_\tau). %\label{eq:single-at-most-alpha}
        \end{align}
        \Cref{item:decomp-single-2} then follows by setting $T_a = e(B_{s, a})$ for each $a \in A_\tau$, and since $\{B_{s, a}\}_{a \in A_\tau}$ forms a partition of $B_s$. 

        \paragraph{Constructing the individual swaps of \cref{item:decomp-swap-2}.}
        We first define the edge swaps $\{o, N_o\}$ for each $o \in e(B_{s, a})$. Consider the matroid $\cM_{s, a}$ contracted on $(v(A_\tau) \setminus N_{s, a}) \cup (v(O) \setminus B)$, and observe that $N_{s, a}$ and $B_{s, a}$ are both independent sets in $\cM_{s, a}$.
        Partition $B_{s, a, o} \triangleq B_{s, a} \cap v(o)$ for $o \in e(B_{s,a})$.
        We apply \Cref{thm:Rota-non-basis-2} with the partition $\{B_{s, a, o}\}_{o \in e(B_{s, a})}$ of $B_{s, a}$ to get a partition of $N_{s, a}$ into $\{N_{s, a, o}\}_{o \in e(B_{s, a})}$ such that $N_{s, a} \setminus N_{s, a, o} \cup B_{s, a, o}$ is independent in $\cM_{s, a}$ for each $o \in e(B_{s, a})$.
        For every $a \in A_\tau$ and $o \in e(B_{s, a})$, we can then set where $N_o \leftarrow N_{s, a, o}$ to conclude \cref{item:decomp-swap-2} for all $o \in O_s$. \\
    
        We repeat the same procedure for $O_d$ and $O_b$. Recall that $\cM'$ is the matroid contracted on $v(O) \setminus B$, and that for each $q \in \{d, b\}$ we have $N_q \subseteq v(A_\tau)$ with the property that $|B_q| = |N_q|$ and such that $(v(A_\tau) \setminus N_q) \cup B_q$ is independent in $\cM'$.
        Let $\cM''$ be the matroid contracted on $(v(A_\tau) \setminus N_q) \cup (v(O) \setminus B)$ and observe that $N_q$ and $B_q$ are independent sets in $\cM''$ of equal size.
        For each $o \in O_q$, we define the sets $B_{q, o} \triangleq B_q \cap v(o)$ and apply \Cref{thm:Rota-non-basis-2} with the partition $\{B_{q, o}\}_{o \in O}$ of $B_q$. This implies that there exists a partition of $N_q$ into $\{N_{q, o}\}_{o \in O}$ such that $N_q \setminus N_{q, o} \cup B_{q, o} \in \cM''$.
        Thus $(v(A_\tau) \setminus N_{q,o}) \cup B_{q, o} \cup (v(O) \setminus B)$ is independent in $\cM$. Observe that $e(B_{q, o}) = \{o\}$ and
        \begin{align*}
            v(A_\tau) \setminus N_{q, o} \cup v(o) & \subseteq v(A_\tau) \setminus N_{q, o} \cup v(o) \cup (v(O) \setminus B) \\
            & \subseteq v(A_\tau) \setminus N_{q, o} \cup B_{q, o} \cup (v(O) \setminus B) \in \cI,
        \end{align*}
        where the last inequality follows from the fact $v(o) = (v(o) \cap B) \cup (v(o) \setminus B)$ and that $v(o) \cap B = B_{q, o}$.
        This concludes \cref{item:decomp-swap-2} by setting $N_o = N_{q, o}$ for each $o \in O_d \cup O_b$. \\

        The proof of \Cref{item:decomp-deg-1} follows from the fact that $|N_o| = |v(o) \cap B|$ for all $o \in O$ by construction.
    \end{proof}

    \section{Polynomial Runtime}
\label{sec:runtime}
    We prove that \Cref{alg:merge-LS} runs in polynomial time. Each iteration of our algorithm requires finding a $(\theta_i, \alpha)$-improvement, where $\alpha > 1$.

    \begin{myprop}
        \label{prop:finding-improvement}
        For any $i \in [L]$, it is possible to find $(\theta_i, \alpha)$-improvement in polynomial time for constant $k$.
    \end{myprop}
    \begin{proof}
        [Proof of \Cref{prop:finding-improvement}]
        One can simply enumerate all sets $S$ of size at most $k$ and $N$ of size $k^2$ as well as all $k!$ orderings of $S$.
    \end{proof}

    \runtime*
    \begin{proof}[Proof of \Cref{lem:runtime-merge-LS}]
        The runtime of \Cref{alg:merge-LS} depends on the runtime of \Cref{alg:local-search} and the number of times the while loop is executed. We will prove that the number of times that the while loop is executed is bounded, this automatically implies the second part of the lemma.
        Observe first that finding a feasible $(\theta, \alpha)$-swap can be done in polynomial time (when $k$ is constant).
        We first claim that the total number of $(\theta, \alpha)$-swaps executed by \Cref{alg:merge-LS} is bounded.
        Let $m = \delta \frac{W}{n}$ be the smallest threshold considered by \Cref{alg:merge-LS}, where $n \triangleq |E|$.
        Consider an improving $(\theta, \alpha)$-swap $(S, N)$ and current solutions $A$ and $A' = A \setminus N \cup S$.
        The proof of \Cref{lem:discard-elements} shows that
        \begin{align*}
            f(A')-f(A)\geq (\alpha-1)\sum_{a\in N} w(a,A) \geq (\alpha - 1)|N|m.
        \end{align*}
        The last inequality follows from the fact that $\theta \geq m$ and that the incremental value $w(a, A)$ is monotonically increasing.
        In the case $N = \emptyset$, we have by submodularity $f(A') - f(A) \geq f(S \mid X \setminus S) \geq \theta \geq m$ by definition of $(m_i, \alpha)$-improvements.
        Therefore, every $(\theta, \alpha)$-improvement increases the objective value by at least $\min\{1, \alpha - 1\} m$.
        We now bound the total number of improvements across rounds. Let $O$ be an optimal solution, and let $W \triangleq \max_{e\in E\colon e \in \cI} f(e)$. Then, 
        \begin{align*}
            f(A) & \leq f(O) \leq \sum_{e \in O} f(e) \leq n W.
        \end{align*}
        Hence, there are at most $\frac{nW}{\min\{1, \alpha - 1\}m} = \delta^{-1} n^2$ improving $\alpha$-swap in total. \\

        This allows us to bound the number of times that the \emph{while-loop} is executed. The while loop is executed after each decrease of $\theta$. After each decrease, we execute $(\alpha, \theta)$-$\LS$.  By the above, the number of times an improvement is found is at most $\delta^{-1} n^2$. Thus, there are at most $\delta^{-1} n^2$ decreases of $\theta$ which correspond to a swap. Finally, observe that $n^{O(k)}$ consecutive threshold-decrease of $\theta$ without any improving swap terminates the algorithm.
        Thus, the runtime of \Cref{alg:merge-LS} is $O(\delta^{-1} n^{O(k)})$. 
    \end{proof}
    
\section{Approximation Guarantee Calculations}
\label{sec: appendix-alpha-calc}
We prove the calculations that give final approximation guarantee restated here. 
\alphacalc*

 \begin{proof}
        [Proof of \Cref{cor:approx-num}]
        Let's expand the terms of \Cref{thm:main-approx}. We have that
        $\alpha/(\alpha - 1) = 1 + 1/(\alpha - 1) = 1 + 1/\alpha + 1/(\alpha(\alpha - 1)) = 1 + \sqrt{\frac{\log(k)}{k}} + O(\frac{\log(k)}{k})$.
        Similarly, we have that $\log(k/\alpha) = \log(\sqrt{k \log(k)}) = 1/2 \log(k \log(k)) = \log(k)/2 + \log(\log(k))/2$.
        In particular, multiplying the previous term by $\alpha$, we get $\alpha \log(k/\alpha) = \sqrt{k \log(k)}/2 + \sqrt{k} \cdot \frac{\log(\log(k))}{2\sqrt{\log(k)}}$.
        Therefore,
        \begin{align*}
            & \quad\frac{\alpha}{\alpha - 1} \lb \frac{k}{2} + 1 + \delta + \alpha + \alpha \log(k/\alpha) \rb \\
            = & \quad\lb 1 + \sqrt{\frac{\log(k)}{k}} + O\lb \frac{\log(k)}{k} \rb \rb \! \lb \frac{k}{2} + 1 + \delta + \sqrt{\frac{k}{\log(k)}} + \frac{\sqrt{k \log(k)}}{2} + O \lb \sqrt{\frac{k}{\log(k)}} \log(\log(k))\! \rb\! \rb \\
        = & \quad k/2 + \sqrt{k \log(k)} + O \lb \sqrt{\frac{k}{\log(k)} \log(\log(k))} \rb \\
        = & \quad k/2 + \sqrt{k \log(k)} + o \lb \sqrt{k \log(k)} \rb
        \end{align*}
        for any choice of constant $\delta$.\qedhere
    \end{proof}

\end{document}